\documentclass[10pt,journal,compsoc]{IEEEtran}
\usepackage{graphics} 
\usepackage{epsfig} 
\usepackage{mathptmx} 
\usepackage{times} 
\usepackage{amsmath} 
\usepackage{amssymb}  
\usepackage{url}
\usepackage{makecell}
\usepackage{amsthm}
\usepackage{algorithm}  
\usepackage{algorithmicx}
\usepackage[noend]{algpseudocode}
\usepackage{multirow}
\usepackage{array}

\ifCLASSOPTIONcompsoc
  \usepackage[nocompress]{cite}
\else
  \usepackage{cite}
\fi

\ifCLASSINFOpdf
\else
\fi

\begin{document}
\newtheorem{definition}{Definition}
\newtheorem{theorem}{Theorem}
\newtheorem{lemma}{Lemma}
\newtheorem{case}{Case}
\renewcommand\arraystretch{1.2}
%
\title{Watermarking in Secure Federated Learning: A Verification Framework Based on Client-Side Backdooring}

\author{Wenyuan~Yang, Shuo~Shao, Yue~Yang, Xiyao~Liu,~\IEEEmembership{Member,~IEEE,} Ximeng~Liu,~\IEEEmembership{Senior Member,~IEEE,} Zhihua~Xia,~\IEEEmembership{Member,~IEEE,} Gerald~Schaefer,~\IEEEmembership{Member,~IEEE,} and~Hui~Fang,~\IEEEmembership{Member,~IEEE,}%
\IEEEcompsocitemizethanks{
\IEEEcompsocthanksitem W. Yang is with the School of Electronics Engineering and Computer Science, Peking University, China. Email address: {\tt ylrftx@gmail.com}%
 \IEEEcompsocthanksitem S. Shao, Y. Yang and X. Liu are with the School of Computer Science and Engineering, Central South University, Changsha 410083, China, Email address: {\tt shaoshuo\_ss@outlook.com, yangy0101@outlook.com, lxyzoewx@csu.edu.cn}%
\IEEEcompsocthanksitem X. Liu is with the College of Computer and Data Science, Fuzhou University, Fuzhou 350116, China, E-mail address: {\tt snbnix@gmail.com}%
\IEEEcompsocthanksitem Z. Xia is with the College of Cyber Security, Jinan University, Guangzhou, 510632, China, E-mail address: {\tt xia\_zhihua@163.com}%
\IEEEcompsocthanksitem G. Schaefer and H. Fang are with the Department of Computer Science, Loughborough University, Loughborough LE11 3TU, UK. E-mail address: {\tt h.fang@lboro.ac.uk}%
\IEEEcompsocthanksitem Xiyao Liu is the corresponding author.}%
}

%
%

\markboth{}%
{Shell \MakeLowercase{\textit{et al.}}: Bare Demo for Computer Society Journals}
%



\IEEEtitleabstractindextext{%

\begin{abstract}
Federated learning (FL) allows multiple participants to collaboratively build deep learning (DL) models without directly sharing data. Consequently, the issue of copyright protection in FL becomes important since unreliable participants may gain access to the jointly trained model. Application of homomorphic encryption (HE) in secure FL framework prevents the central server from accessing plaintext models. Thus, it is no longer feasible to embed the watermark at the central server using existing watermarking schemes. In this paper, we propose a novel client-side FL watermarking scheme to tackle the copyright protection issue in secure FL with HE. To our best knowledge, it is the first scheme to embed the watermark to models under the Secure FL environment.  We design a black-box watermarking scheme based on client-side backdooring to embed a pre-designed trigger set into an FL model by a gradient-enhanced embedding method. Additionally, we propose a trigger set construction mechanism to ensure the watermark cannot be forged. Experimental results demonstrate that our proposed scheme delivers outstanding protection performance and robustness against various watermark removal attacks and ambiguity attack.
\end{abstract}

\begin{IEEEkeywords}
Federated learning, copyright protection, digital watermark, client-side backdooring.
\end{IEEEkeywords}}

\maketitle

\IEEEdisplaynontitleabstractindextext

%
\IEEEpeerreviewmaketitle

\IEEEraisesectionheading{\section{Introduction}\label{Sec: Introduction}}

%
%
%
%
\IEEEPARstart{F}{ederated} learning (FL)~\cite{mcmahan2017communication}, which enables multiple data owners to learn a machine learning (ML) or deep learning (DL) model with joint efforts, is increasingly used in various applications such as medical image analysis~\cite{adnan2022federated, ng2021federated, kumar2021blockchain}, word predictions~\cite{hard2018federated, zhu2020empirical, singhal2021federated} and recommendation systems~\cite{yang2020federated, muhammad2020fedfast}. In this manner, a large amount of private data from multiple participants is available to generate accurate and reliable ML models. However, FL protects participants' personal data privacy at the sacrifice of holding the model also privately. Compared with traditional centralised model training, all parties in FL have access to the global model, increasing the risk of model leakage. Since training an FL model requires a large number of clients and computations, copyright protection of FL models becomes an important issue.


Model watermarking methods, which are currently used to protect the copyright of DL models, can be employed also to protect FL models. Broadly, current watermarking schemes can be divided into two categories: white-box watermarks and black-box watermarks~\cite{regazzoni2021protecting, fkirin2022copyright}. White-box approaches embed the digital watermark directly into the parameters of the DL model, and require the whole model for verification~\cite{uchida2017embedding, wang2021riga, szyller2021dawn, maini2021dataset}. Black-box schemes use backdoor attacks~\cite{gu2017badnets}, generating a unique set of data named a trigger set to embed the watermark into DL models~\cite{adi2018turning, zhang2018protecting, merrer2020adversarial, xue2022active, shafieinejad2021robustness}. Compared with white-box watermarking, black-box techniques require the application programming interface (API) of a DL model instead of direct access, making them ideal for copyright verification of both DL and FL models.

WAFFLE~\cite{tekgul2021waffle}, a recently proposed black-box watermarking scheme designed for FL, embeds the watermark into the global model by adding a retraining step at central server. After the aggregation phase, the central server takes trigger set and retrains the global model to embed the watermark. However, WAFFLE is not effective in secure federated learning, a framework that uses cryptographic methods to further enhance the protection of privacy-sensitive data and prevent privacy leakage from gradients~\cite{aono2017privacy}. Specifically, the secure FL based on homomorphic encryption (HE)~\cite{paillier1999public}, which encrypts the gradients, prevents the central server from accessing the plaintext of model parameters. Consequently, it is no longer feasible to directly use watermarking methods at the central server. In addition, WAFFLE is built on the assumption that the central server is the initiator of the FL procedure, i.e., that it is the owner of the FL model. This it is not always the case in business-to-consumer (B2C) FL~\cite{yang2019federated}, where an initiator gets its clients to jointly train an FL model. In this case, the central server can be a hired third party different from the initiator and other ordinary participants, while the initiator (a sole participant or a subgroup of participants) should be the actual owner(s). 

In this paper, we propose a novel watermarking scheme to tackle the issue of copyright protection in secure FL. Specifically, we present an FL model copyright protection approach for the initiator by embedding a backdoor into the FL model from the client side. This client-side watermarking scheme overcomes the limitation on embedding due to HE. Moreover, we design a unique enhancement method, 
and propose a novel trigger set construction method using a permutation-based secret key to tackle the problem of ambiguity attacks. 

Our contributions in this paper are:
\begin{itemize}
    \item A novel watermarking scheme is proposed to protect the copyright of secure FL models by embedding backdoor-based watermarks into FL models client-side, overcoming the limitations on watermark embedding due to HE encryption of the model.
    \item A non-ambiguous trigger set construction mechanism is designed for watermark embedding based on a permutation-based secret key and noise-based patterns to prevent adversaries from forging the watermark.
    \item A gradient-enhanced watermark embedding method is deployed to tackle the issue of slim effects of single clients on watermark embedding.
    \item Comprehensive experiments demonstrate that our method is resilient to watermark removal attacks, including fine-tuning, pruning, quantisation, and pattern embedding and spatial-level transformation.
\end{itemize}

The structure of the remainder of the paper is organized as follows. Section~\ref{Sec:Preliminaries} covers some necessary preliminaries, while the black-box watermarking problem in secure FL is formulated in Section~\ref{Sec:ProblemFormulation}. We introduce our proposed method in detail in Section~\ref{Sec:Watermarking}. A security analysis is conducted in Section~\ref{Sec:SecurityAnalysis}, while experimental results are presented in Section~\ref{Sec:Evaluation}. Finally, Section~\ref{Sec:Conclusions} concludes the paper.

\section{Preliminaries}
\label{Sec:Preliminaries}

In this section, we review three related preliminaries, including backdoor attacks, black-box watermarking in Machine Learning and homomorphic encryption scheme. The symbols used throughout the paper are summarised in Table~\ref{Table:Glossary}.

\begin{table}[h]
    \caption{Symbol definitions}
    \label{Table:Glossary}
    \centering
    \begin{tabular}{cl}
    \hline
    symbol & definition \\
    \hline
    $I$ & input space of DL models\\
    $O$ & output space of DL models\\
    $\mathbb{D}$ & Cartesian production of $I$ and $O$, i.e., $\mathbb{D}=I\times O$ \\
    $\hat{M}$ & the watermarked model \\
    $D_s$ & trigger set \\
    $D_b$ & dataset with benign samples \\
    $\mu, \nu$ & patch parameters \\
    $lk$ & location key  \\
    $ck$ & classification key \\
    
    $M_j$ & weights of model in $j$-th iteration \\
    $d$ & dimensionality of model parameters \\
    $L_j^i$ & local gradients of $i$-th client in $j$-th iteration \\
    $D^i$ & private dataset of $i$-th client \\
    $G_j$ & aggregated global gradients in $j$-th iteration \\
    $\lambda$ & scaling factor for watermark embedding \\
    \hline
    \end{tabular}
    \end{table}





\subsection{Backdoor Attack}
\label{Sec:BackdoorAttack}
Backdoor attack is a special technique which trains an ML model to make predesigned incorrect predictions when encountering some specific inputs~\cite{gu2017badnets, liu2020reflection, bagdasaryan2020backdoor, rieger2022deepsight}. The set of these specific inputs is called the trigger set.

\begin{definition}[Trigger Set]
\label{Def:TriggerSet}
The trigger set is the dataset with inappropriate labels during backdoor attack training. Let $I$ and $O$ be the input space and output space, respectively. Then, $T \subset \mathbb{D}$, $\mathbb{D}=I \times O$ is the trigger set generated from input and output spaces. Any element $(x,y) \in T$ should satisfy
\begin{equation}
\label{Eq:TriggerSe Definition}
f(x) \neq y ,
\end{equation}
where $x \in I$ are the input samples, $y \in O$ is the output, $f:I \rightarrow O\cup\{\perp\}$ is the function that outputs the ground truth label of input sample, and $\perp$ indicates that the ground truth label is not defined in the task. While benign samples $(x_b, y_b)$ has $f(x_b)=y_b$.
\end{definition}

A successfully backdoored ML model will not only output incorrect answers for the trigger set but also still perform well for benign inputs. 
Compared with the ground truth function $f$, the backdoor-attacked function $f^*$ satisfies
\begin{equation}
    \begin{aligned}
        &\mathrm{Pr}[f(x) \neq f^*(x) \mid x \in O \setminus T] \leq negl(\kappa) \\
        \vee \enspace &\mathrm{Pr}[f(x) = f^*(x) \mid x \in T] \leq negl(\kappa) ,
    \end{aligned}
\end{equation}
where $negl(\kappa)$ is a negligible function, and $\kappa$ is a security parameter.

\begin{definition}[Strong Backdoor~\cite{adi2018turning}]
A backdoor is strong if it is hard to be removed even for someone with full knowledge of the exact key generation and trigger set construction algorithms.
\end{definition}

Strong backdoors are related to the robustness of backdoor-based watermarks, and will be further discussed in Section~\ref{Sec:RobustnessAnalysis}.

\subsection{Black-box Watermarking in Machine Learning}

Black-box watermarking uses backdoor attacks to force an ML model to remember specific patterns or features~\cite{adi2018turning, guo2018watermarking, li2019how}. The backdoor attack causes the ML model to misclassify when encountering samples in the trigger set. The model owner keeps the trigger set secret and can thus verify the ownership of the model by triggering a misclassification.

A general black-box watermarking scheme for ML models can be split into three stages: trigger set construction, embedding, and verification.



\subsubsection{Trigger Set Construction}
In black-box watermarking, the trigger set is the direct carrier of the watermark which is embedded into the watermarked model. A trigger set construction algorithm $TrigCons:\varnothing \rightarrow \mathbb{D^*}$ generates the trigger set $D_s$. $D_s$ should be kept secret by the owner.


\subsubsection{Embedding}
In the embedding phase, the owner of the model uses a particular algorithm to train the model so that the model will have a specific output for the trigger set. An embedding algorithm $Embedding: \mathbb{R}^d \times \mathbb{D}^{|D_b|} \times \mathbb{D}^{|D_s|} \rightarrow \mathbb{R}^d$
creates the watermarked model $\hat{M}$ from the original model $M\in \mathbb{R}^d$, where $d$ is the dimensionality of model parameters, the trigger set $D_s$ and a benign dataset $D_b$.

\subsubsection{Verification}
A verification algorithm $Verify: \mathbb{R}^d \times \mathbb{D}^{|D_s|} \rightarrow \{0, 1\}$ takes the watermarked model and the secret dataset held by the owner and outputs the result of verification. `$1$' represents that the model owner successfully confirms their copyright of the model, while `$0$' represents the opposite. In black-box settings, the owner only needs to know the classification results of specific samples instead of the exact model weights. Thus, it can be carried out using API queries. 


\subsection{Homomorphic Encryption}
\label{Sec:CKKSHomomorphicEncryption}
Homomorphic encryption (HE)~\cite{rivest1978data, paillier1999public} is a cryptographic scheme that enables some specific mathematic operations, e.g., addition or multiplication, via encrypted data without decryption. HE thus allows processing of data without sharing the plaintext and HE is consequently employed to provide secure transmission and aggregation in secure FL~\cite{yang2019federatedlearning, park2022privacy, ma2022privacy}.

\section{Problem Formulation}
\label{Sec:ProblemFormulation}

\subsection{System Model}
As shown in Fig.~\ref{Fig:SystemModel}, three different parties are involved in watermarking in secure FL: the initiator, the central server, and ordinary clients. All three parties need to perform Secure FL and jointly train a DL model.

\begin{figure}[t]
    \centering
    \includegraphics[width=.95\linewidth]{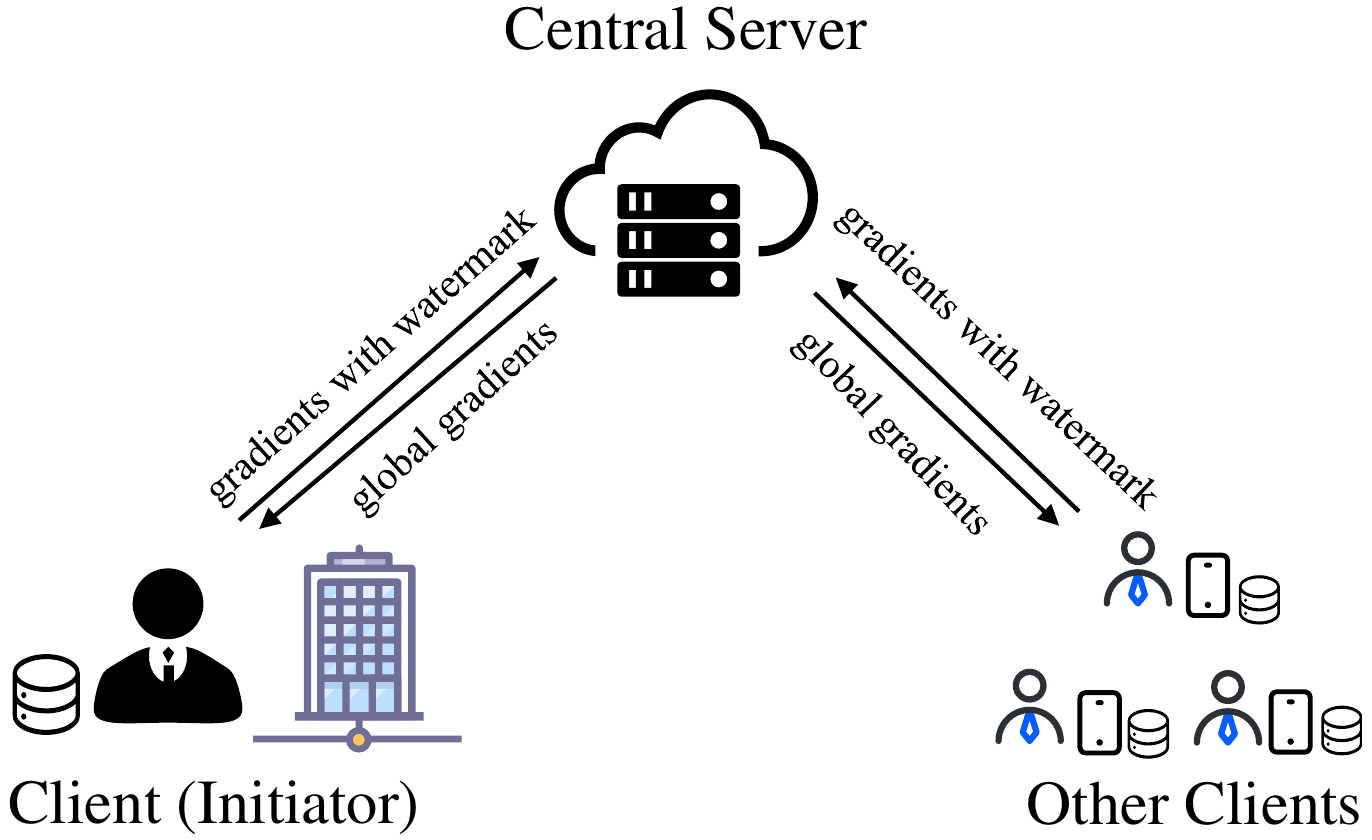}
    \caption{System model of watermarking in secure FL.}
    \label{Fig:SystemModel}
    \vspace{-10pt}
\end{figure}

The \textbf{clients} are the data owner in FL. Each client holds a bunch of personal data. When training an FL model, clients obtain the global gradients from the central server during each iteration. They then update their local model based on the global gradients and train it for several epochs. Finally, they send their local gradients to the central server for aggregation. To avoid privacy leakage, the gradients are encrypted by HE.

The \textbf{initiator} is a special client who is chosen to embed the watermark. Besides training a local model, the initiator is responsible for embedding the black-box watermark into the global model. In case the watermarked model is stolen by an adversary, the watermark can provide copyright verification for the initiator.

The \textbf{central server} is responsible for aggregating the local gradients. In each iteration, the central server collects the encrypted local gradients and aggregates them using an aggregation method. After aggregation, the central server sends the ciphertext of gradients to each client.

\subsection{Threat Model}
\label{Sec:ThreatModel}
In the threat model, we follow the prescriptive semi-honest assumption~\cite{yang2019federated}, which means any clients, including adversaries, should follow the pre-designed secure FL procedure, but might copy and steal the global model. When an adversary steals the watermarked model, they try to remove the watermark or forge their own trigger set. A benign user might also do some processing to the model before deployment. The model may thus be exposed to a variety of attacks.

\begin{definition}[Fine-tuning Attack]
A fine-tuning attack refers to the attempt of removing the watermark from the watermarked model by training the pre-trained model for a few iterations using a new dataset and a small learning rate. A fine-tuning attack algorithm $FineTune:\mathbb{R}^d \times \mathbb{D}^{|D_n|} \rightarrow \mathbb{R}^d$ can be formally defined as
\begin{equation}
    \label{Eq:Fine-tune}
    M_{ft} = FineTune(\hat{M}, D_n, \eta_{ft}) ,
\end{equation} 
where $M_{ft}$ is the model after fine-tuning, $D_n$ is the dataset used for fine-tuning, and $\eta_{ft}$ is the learning rate.
\end{definition}

\begin{definition}[Pruning Attack]
 A pruning attack uses a pruning method to remove the watermark from the watermarked model, where the pruning aims to decrease the number of effective parameters, e.g.\ by setting some unimportant parameters to zero. 
 A pruning attack algorithm $Prune:\mathbb{R}^d \rightarrow \mathbb{R}^d$ which can be defined as
\begin{equation}
    \label{Eq:PruningAttack}
    M_p = Prune(\hat{M}) ,
\end{equation}
where $M_p$ are the pruned model parameters.
\end{definition}

\begin{definition}[Quantisation Attack]
A quantisation attack attempts to remove the watermark through quantisation. Quantisation does not change the concrete structure of the model, but tries to reduce the number of bits used to represent each parameter. Let $b$ be the number of bits before quantisation, while $b'$ bits are used after, with $b > b'$. A quantisation algorithm $Quantise: \{\{0, 1\}^b\}^d \rightarrow \{\{0, 1\}^{b'}\}^d$  can be defined as 
\begin{equation}
    \label{Eq:QuantisationAttack}
    M_q = Quantise(\hat{M}) ,
\end{equation}
where $M_q$ are the quantised model parameters.
\end{definition}

\begin{definition}[PST Attack]
A pattern embedding and spatial-level transformation (PST)  attack~\cite{guo2021fine} attempts to invalidate a backdoor-based watermarking scheme through pre-processing input images in order to affect the classification of the secret trigger set. PST first resizes the input images, then uses a median filter to process some rows and columns of the images, and finally applies a spatial-level transformation to further process the images with random affine transformation and elastic distortions.
\end{definition}

\begin{definition}[Ambiguity Attack]
    In addition to removing the existing watermark from the watermarked model, attempting to forge a fake watermark is another common way to attack a watermarked model.
 An ambiguity attack tries to forge a watermark and verify the ownership of the existing watermarked model. An ambiguity attack algorithm $Forge: \mathbb{R}^d \rightarrow \mathbb{D}^{|D_s|}$ can be defined as
\begin{equation}
   \label{Eq:AmbiguityAttack}
   D_s^{'} = Forge(\hat{M}) ,
\end{equation}
where $\hat{M}$ is the watermarked model and $D_s^{'}$ is the fake trigger set. When the attacker uses $D_s'$ and $\hat{M}$ to claim its copyright, it can cause ambiguity and harm the actual owner.
\end{definition}

\begin{figure*}[t!]
    \centering
    \includegraphics[width=0.95\linewidth]{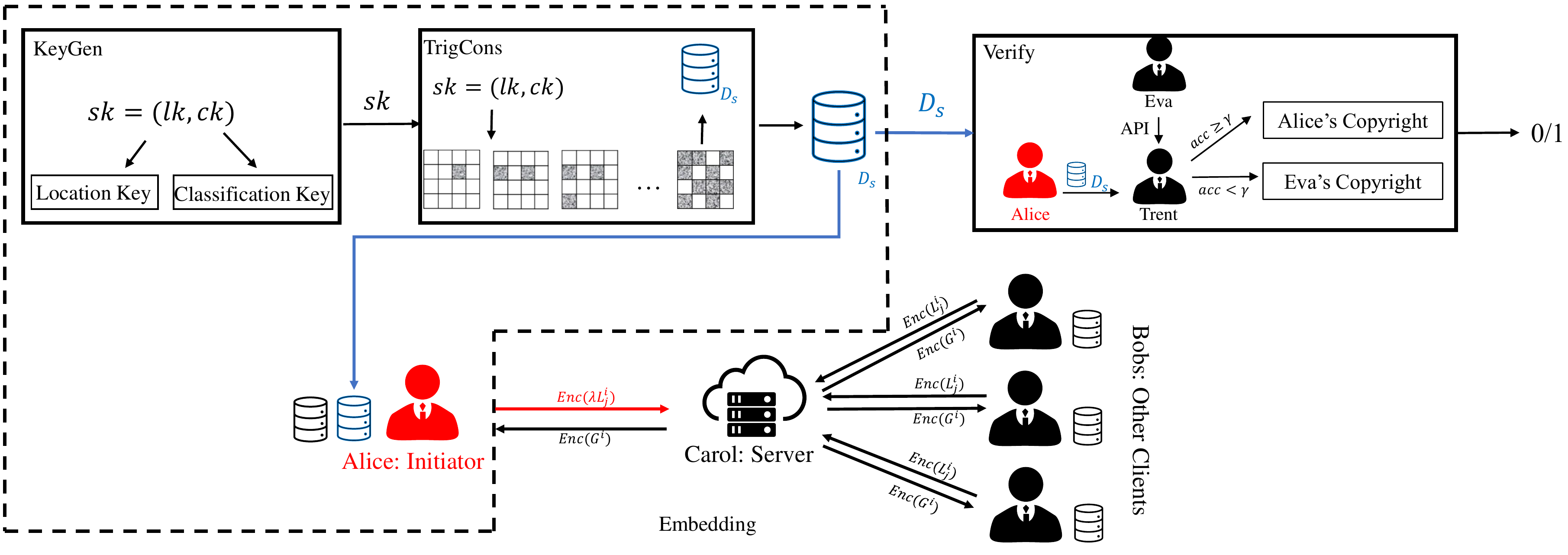}
    \caption{Framework of our proposed watermarking scheme.}
    \label{Fig:Overview}
    \vspace{-10pt}
\end{figure*}

\subsection{Design Goals}
\label{Sec:DesignGoals}
While designing a black-box watermarking scheme for FL, it is important to satisfy the following properties. In general, an outstanding watermarking scheme should provide \textbf{\emph{effectiveness, function preservation, low false-positive rate, robustness}}, and \textbf{\emph{non-ambiguity}}.
\begin{itemize}
\item
\textbf{Effectiveness}: Effectiveness signifies that if the tested model is actually the model embedded with the watermark, the verification algorithm will always output `$1$' when the input is an element from the trigger set $D_s$. This can be formally defined as
\begin{equation}
\label{Eq:CorrectVerification}
\mathrm{Pr}[Verify(\hat{M}, D_s)=1 \mid \hat{M}=Embedding(M, D_b, D_s)]=1 .
\end{equation}
\item 
\textbf{Function Preservation}: Function Preservation refers to that the watermarked model performs approximately as well as the primitive model. It indicates that the watermarking scheme has negligible impact on the functionality of the model. This can be reflected by the model's accuracy on the validation set $D_v$ and can be defined as
\begin{equation}
\label{Eq:FunctionPreservation}
Acc(M, D_v) - Acc(\hat{M}, D_v) \leq negl(\kappa) ,
\end{equation}
where $Acc:\mathbb{R}^d \times \mathbb{D}^{*} \rightarrow \mathbb{R}$ denotes the accuracy on validation set.
\item
\textbf{Low False Positive Rate}: It is crucial that the watermark should not be extracted from an unwatermarked model, i.e., the watermarking scheme should have a low false positive rate. Generally, for a $k$-class model, the false positive rate should be near or below $1/k$, which is the probability of guessing. Formally defined, for a given secret trigger set $D_s$ and a model $M$ without watermark,
\begin{equation}
\label{Eq:LowFalse-postiveRate}
Acc(M, D_s) \leq 1/k \enspace \vee \enspace Acc(M, D_s) \approx 1/k .
\end{equation}
\item
\textbf{Robustness}: Robustness means that when an attacker applies any attack algorithm $Att:\mathbb{R}^d \times S \rightarrow \mathbb{R}^d $ to the watermarked model, with $S$ the set of parameters used in the attack, the model should maintain the watermark.
A watermark is robust if one of the following two cases is true after attacks:
\begin{case}
\label{Prop:Robustness1}
Case~\ref{Prop:Robustness1} signifies that the attack does not significantly influence the functionality of the model in both the primitive task and the watermark. The model owner can thus still successfully verify its ownership. This can be formally defined as Eq.~(\ref{Eq:RobustnessProp1}).
\begin{equation}
\label{Eq:RobustnessProp1}
\begin{aligned}
    &Acc(\hat{M}, D_b) - Acc(Att(\hat{M}), D_b) \leq negl(\kappa)\enspace \\
    \wedge \enspace&Verify(Att(\hat{M}), D_s)=1 ,
\end{aligned}
\end{equation}
where $D_b$ is the benign dataset, $D_s$ is the trigger set, and $\hat{M}$ is the watermarked model.
\end{case}
\begin{case}
\label{Prop:Robustness2}
Case~\ref{Prop:Robustness2} signifies that the functionality of the attacked model drops significantly compared to the primitive model and is no longer useful for its task, rendering the attack unsuccessful. This can be formally defined as Eq.~(\ref{Eq:RobustnessProp2}).
\begin{equation}
\label{Eq:RobustnessProp2}
Acc(\hat{M}, D_b) - Acc(Att(\hat{M}), D_b) > negl(\kappa)
\end{equation}
\end{case}
\item
\textbf{Non-ambiguity}: A watermarked model is non-ambiguous if for any ambiguity attack algorithm $Forge:\mathbb{R}^d \rightarrow \mathbb{D}^{|D_s|}$ the probability of successful verification is negligible, i.e.
\begin{equation}
\label{Eq:Non-ambiguity}
\mathrm{P}\mathbf{r}[Verify(\hat{M}, D_s')|D_s'=Forge(\hat{M})] \leq negl(\kappa) ,
\end{equation}
where $D_s^{'}$ is the fake trigger set created by the attacker. Thus, the watermark is considered unforgeable.
\end{itemize}




\begin{figure*}[t]
    \centering
    \includegraphics[width=.95\linewidth]{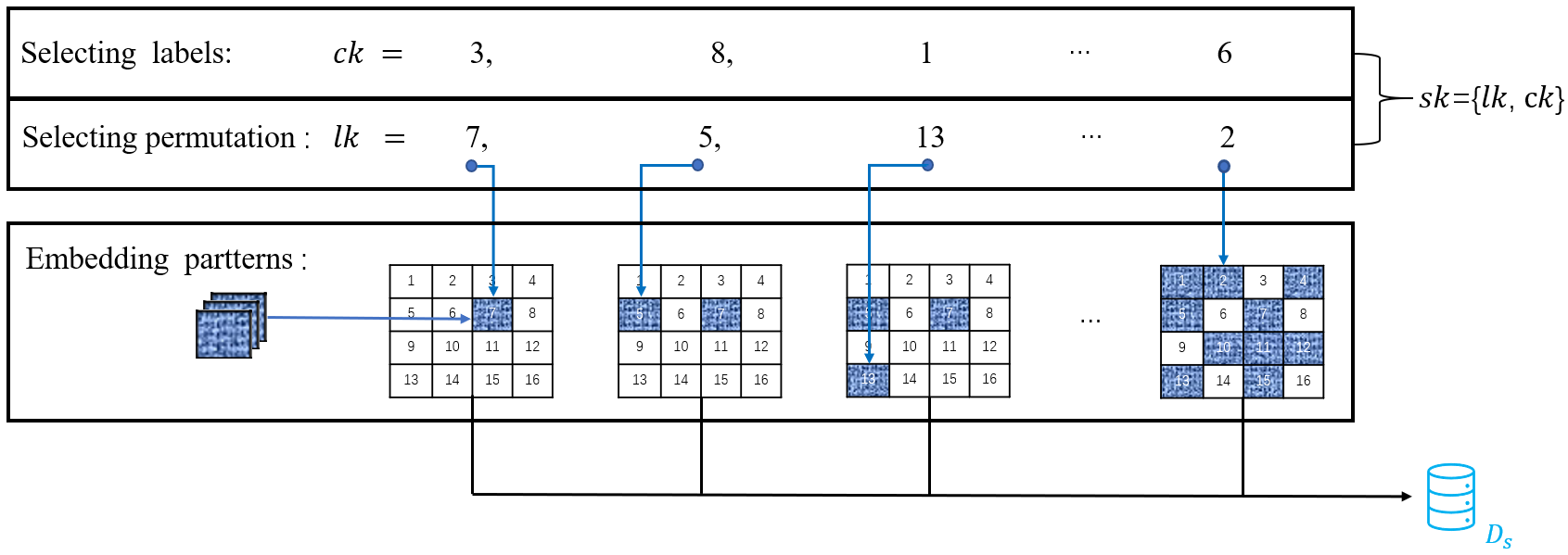}
    \caption{Example of key generation and trigger set construction algorithms.}
    \label{Fig:Keygen}
    \vspace{-10pt}
\end{figure*}

\section{Watermarking in Secure Federated Learning}
\label{Sec:Watermarking}

In this paper, we propose a novel client-side watermarking scheme for secure FL, which includes non-ambiguous key generation and trigger set construction algorithms as well as effective embedding and verification methods. An overview of our approach is shown in Fig.~\ref{Fig:Overview}.

Consider the following scenario: Alice is an enterprise that has a number of Bobs to use its services. Alice decides to improve the quality of its services by training a high-performance DL-based model, however does not own sufficient data to train such a model successfully. Alice therefore decides to use federated learning to generate a model with joint efforts from Bobs. At the training stage, Alice hires Carol to perform as the central server of the FL and utilizes HE to protect the privacy of its clients' local data. Consequently, none of these parties will possibly get any private information from other parties.

Since training an FL model requires sending the model parameters to each client, a malicious client, Eva, may copy the model, although it does not belong to her. To protect the copyright of the model, Alice therefore wants to embed a watermark into the model to allow for verification of unauthorised use of the model via a trustworthy third party Trent. Alice is one of the client-side nodes in the FL system in this scenario, and we follow the prescriptive semi-honest assumption~\cite{yang2019federated} where, except for the initiator Alice, any other party will follow the pre-designed algorithm.

\subsection{Key Generation for Trigger Set}
\label{Sec:KeyGeneration}

In our proposed scheme, a key generation algorithm is utilized to help construct the trigger set. A key generation algorithm $KeyGen:\varnothing \rightarrow \{0, 1\}^*$ randomly generates a bit string $sk$ as the secret key, although $sk$ can also be carefully chosen by the model owner. Non-repudiation and unforgeability of $sk$ directly affect the watermark's robustness against ambiguity attacks. A watermark without $sk$ can be easily forged and thus the watermarking scheme is unfeasible. 

Our insight of key generation and trigger set construction algorithms is based on the following observation. The primary dilemma of key generation and trigger set construction is that, on the one hand, intuitively, making the model remember unlabeled data such as random noise does less harm to the model's functionality than misclassifying meaningful samples. On the other hand, since constructing a trigger set with only one class leads to a security problem as it will be easy to forge, a trigger set with multiple classes is required, while simply forcing a model to classify samples with similar noise to different classes is difficult and may result in overfitting. We therefore first divide an image into several patches and add to the patches to construct a multi-classes trigger set so that the model will remember the location of noise instead of a specific noise pattern.

In the key generation phase, the initiator generates, using the algorithm defined in Algorithm~\ref{Algo:KeyGeneration}, a secret key used to construct the trigger set. In our work, we design the secret key $sk= (lk, ck)$ to comprise two parts. The first part denotes the positions of noise, while the other part denotes the corresponding labels. Assume that we try to watermark a $k$-class FL model. The owner first needs to choose two patch parameters, $\mu$ and $\nu$, that satisfy $\mu\nu \geq k$. 
 The owner then generates a random permutation of $k$ numbers from $\{x \mid x \in \mathbb{N} \wedge x < \mu\nu\}$, where $\mathbb{N}$ is the set of all non-negative integers. The permutation, called location key $lk$, is the first part of $sk$. Since the second part are the labels, the owner generates another permutation of $k$ numbers from $\{x \mid x \in \mathbb{N} \wedge x < k\}$ to yield the classification key $ck$.

\begin{algorithm}[h]
\caption{Key generation algorithm $KeyGen$}
\label{Algo:KeyGeneration}
\hspace*{0.02in}{\bf Input:} number of classes $k$\\
\hspace*{0.02in}{\bf Output:} secret key $sk$
\begin{algorithmic}[1]
\State Choose a pair of integers $\mu$, $\nu$ where $\mu\nu \geq k$
\State Randomly and successively select $k$ numbers from $\{x \mid x \in \mathbb{N} \wedge x < \mu\nu\}$ to construct a permutation to yield location key $lk$
\State Randomly and successively select $k$ numbers from $\{x \mid x \in \mathbb{N} \wedge x < k\}$ to construct a permutation to yield classification key $ck$
\State Set $sk = (lk, ck)$
\State \Return $sk$
\end{algorithmic}
\end{algorithm}

\subsection{Trigger Set Construction}
\label{Sec:TriggerSetConstruction}
After key generation, the initiator uses $sk=(lk, ck)$ to construct the trigger set following the algorithm defined in Algorithm~\ref{Algo:TriggerSetConstruction}. The input images with size $\varphi \times \xi$ are divided into $\mu \times \nu$ patches. In this way, each patch corresponds to an integer between 0 and $\mu\nu - 1$. Then the owner randomly samples $t$ patterns of $\lfloor \varphi / \mu \rfloor \times \lfloor \xi / \nu \rfloor$ pixels. For this, we use Gaussian noise to generate patterns, although the patterns can be any images that depend on the generation algorithm. Each pattern needs to be filled into the specific patch represented by $lk$ and the corresponding label $ck$. 

For example, the first element of $lk$ is $lk_0$ (we use subscript $n$ to denote the $n$-th element of permutation). We thus find the $lk_0$-th patch and fill it with $t$ sampled patterns. Pixel values of other patches which have not been filled should be set to zero. Therefore, we get $t$ different images with only one patch that is non-zero. These $t$ images are labeled $ck_0$. The second element corresponds to the $lk_1$-th patch. Based on the images with $lk_0$-th patch filled, we fill another patch, identified by $lk_1$, with the same pattern. After that, we get another $t$ images filled with two patches which we label $ck_1$. The rest can be done in the same manner. In $l-$th step, we should fill the patches corresponding to $\{lk_0, lk_1, \cdots, lk_{l-1}\}$ with the $t$ patterns and label them with class $ck_{l-1}$. Eventually, we will get $kt$ images with all $k$ classes, and the trigger set is composed of all these images. An example of our key generation and trigger set construction algorithms is illustrated in Fig.~\ref{Fig:Keygen}.

\begin{algorithm}[h]
\caption{Trigger set construction algorithm $TrigCons$}
\label{Algo:TriggerSetConstruction}
\hspace*{0.02in}{\bf Input:} patch parameters $\mu,\nu$; secret key $sk=(lk,ck)$\\
\hspace*{0.02in}{\bf Output:} trigger set $D_s$
    \begin{algorithmic}[1]
        \State Divide input $\varphi \times \xi$ image into $\mu \times \nu$ patches of $\lfloor \varphi / \mu \rfloor \times \lfloor \xi / \nu \rfloor$ pixels
        \State Generate $t$ patterns $P_t$ of $\lfloor \varphi / \mu \rfloor \times \lfloor \xi / \nu \rfloor$ pixels from Gaussian distribution
        \State $D_s = \varnothing$
        \For{$l=1$ to $k$}
            \State Select first $l$ terms of location key $lk$, i.e., $\{lk_0,lk_1,lk_2,$ $\dots, lk_{l-1}\}$
            \State Find set of patches corresponding with $\{lk_0, lk_1, lk_2,$ $ \dots, lk_{l-1}\}$.
            \State $D_l = \varnothing$
            \For{each pattern $p \in P_t$}
                \State Fill corresponding patches with $p$
                \State Add image to $D_l$.
            \EndFor
            \State Label images with $ck_{l-1}$ and get dataset $D_l$
            \State Set $D_s = D_s \cup D_l$
        \EndFor
        \State \Return $D_s$
    \end{algorithmic}
\end{algorithm}

Since $\mu$ and $\nu$ are used to divide the input image space into patches so that the secret key $sk$ can be embedded into the trigger set images, there are some restrictions when choosing the parameters. First, since we need to replace $k$ patches with the generated patterns, there should be at least $k$ patches. Second, each pattern needs to have at least one pixel for embedding. Therefore, the number of patches should not exceed the number of pixels in the input image. For $\varphi \times \xi$ input images, this leads to
\begin{equation}
    k \leq \mu\nu \leq \varphi\xi .
\end{equation}

Since $\mu\nu$ is the upper bound for the location key, it directly determines the size of the location key space; the larger $\mu\nu$, the bigger the location key space and the more difficult for an attacker to forge the secret key. On the other hand, a larger value of $\mu\nu$ yields fewer pixels for one patch, making it more challenging for the FL model to learn the trigger set images and the model potentially overfitting the trigger set. Thus, the robustness of the watermark is not guaranteed. Consequently, there is a trade-off between robustness and security that needs to be considered when choosing the patch parameters.

\subsection{Watermark Embedding}
\label{Sec:WatermarkEmbedding}

\begin{figure}[b!]
\centering
\includegraphics[width=.95\linewidth]{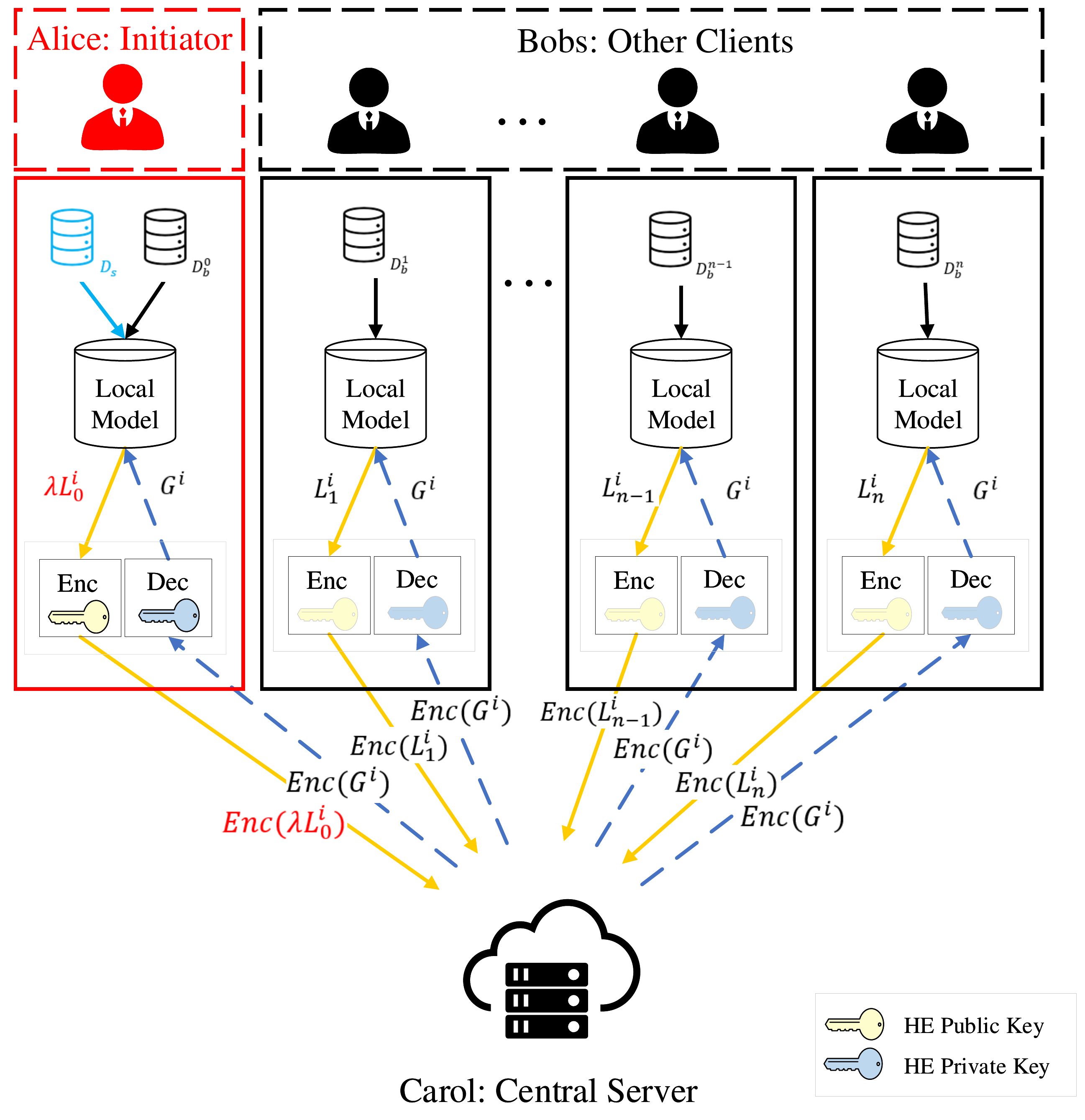}
\caption{Watermark embedding process in secure FL.}
\label{Fig:Embedding}
\vspace{-10pt}
\end{figure}

As illustrated in Fig.~\ref{Fig:Embedding}, we propose a gradient-enhanced algorithm for the initiator to embed the watermark into the FL model. The initiator only needs to perform a subtle change to carry out the embedding when transmitting the gradients. After key generation and trigger set construction, the initiator adds the trigger set into its normal training data samples to calculate the local gradients and embed the backdoor watermark into the jointly-trained model. Since the initiator is not guaranteed to be selected at each iteration and the initiator is only a single node in the FL model, we define a scaling factor $\lambda$ to improve the embedding process. When the initiator is selected to participate in the model update, $\lambda$ is used to enhance the gradient so that the FL model will remember the trigger set. Assuming that the initiator is the $1$-st client in FL, i.e., the index of the initiator is $0$, the clients get the encrypted messages as
\begin{equation}
\label{Eq:Watermark-embedding}
C_j^i=\begin{cases}
        Enc(\lambda L_j^i) &\text{if }i=0 \\
        Enc(L_j^i) &\text{otherwise}
\end{cases} ,
\end{equation}
where $L_j^i$ represents the local gradients of the $i$-th client in the $j$-th iteration, $C_j^i$ is the corresponding ciphertext, and $Enc$ represents any homomorphic encryption algorithm.

The watermark-embedding secure transmission algorithm is defined in Algorithm~\ref{Algo:Watermark-embedding}.

\begin{algorithm}[h!]
\caption{Watermark-embedding secure transmission algorithm}
\label{Algo:Watermark-embedding}
\hspace*{0.02in}{\bf Input:} local gradients $L_j^i$; scaling factor $\lambda$\\
\hspace*{0.02in}{\bf Output:} secure gradient $C_j^i$
    \begin{algorithmic}[1]
        \If{$i==0$} // the initiator node
            \State $C_j^i=Enc(\lambda L^i_j)$
        \Else
            \State $C_j^i=Enc(L^i_j)$
        \EndIf
        \State \Return $C_j^i$
    \end{algorithmic}
\end{algorithm}

The scaling factor $\lambda$ plays a significant role in the embedding and directly influences the watermarking scheme's effectiveness and function preservation quality. A small $\lambda$ results in a small effect of the initiator which might leads to a failed embedding, while a large $\lambda$ might significantly influence the global model's functionality. We design a suitable scaling factor setting as
\begin{equation}
\label{Eq:ScalingFactor}
\lambda = \frac{N}{n},
\end{equation}
where $n$ is the number of clients selected in each iteration, and $N$ is the total number of clients. For a small local learning rate which is common in DL, the model parameters change consistently in the training procedure. The gradients of each client during several iterations are approximately equal. Therefore, this is approximately equivalent to the case of the initiator node being selected in every iteration, corresponding to a well-applied technique in existing strong backdoor attacks~\cite{gu2017badnets, adi2018turning}. Our method thus has a consistent updating process to improve the stability of model convergence while performing well for watermarking.

\subsection{Watermark Verification}
When Alice needs to verify an unauthorised model deployment by Eva, Alice and Eva can recruit a fully-trustworthy third party Trent as the arbitrator who employs the watermark verification protocol define in Algorithm~\ref{Algo:Verification}.

\begin{algorithm}[h!]
\caption{Watermark verification protocol}
\label{Algo:Verification}
\hspace*{0.02in}{\bf Input:} a subset $D \subset D_s$; the API of Eva's model $\hat{M}$\\
\hspace*{0.02in}{\bf Output:} boolean value indicating verification result
    \begin{algorithmic}[1]
        \State Get $D$
        \State Get API of model $\hat{M}$
        \State Calculate accuracy $acc = Acc(\hat{M}, D)$
        \If{$acc \geq \gamma$}
            \State \Return $1$ // copyright verified
        \Else
            \State \Return $0$ // copyright unverified
        \EndIf
    \end{algorithmic}
\end{algorithm}

Alice first sends a subset of the secret trigger set $D_s$ to Trent. Trent then gets the API of Eva's model so that he can input samples and receive classification results. Trent uses the API to process Alice's samples and checks whether the results are in accord with the labels provided by Alice. If the resulting accuracy is above a threshold $\gamma$, the Alice's ownership of the model is confirmed and Trent can sentence Eva for infringement.

The threshold $\gamma$ controls the probability of a false positive verification. If a $k$-class model does not have a backdoor embedded watermark, it will classify samples from the backdoor dataset to random labels with a probability of correct classification of $1/k$. We expect the probability of successful copyright verification to be negligible on the model without backdoor. Assuming that Alice provides $n$ backdoor samples, threshold $\gamma$ satisfies
\begin{equation}
        \sum_{d=\lceil\gamma n\rceil}^n\tbinom{n}{d}(\frac{1}{k})^d(\frac{k-1}{k})^{n-d}\leq negl(\kappa) .
\end{equation}

Our protocol has some distinct advantages. First, Trent does not need to access the weights of Eva's model, avoiding the possibility of Eva cheating Trent by providing fake model parameters. Since the model is deployed to all the clients via API access, it is impossible for Eva to avoid scrutiny. Second, Alice needs to provide only some samples of her trigger set, allowing to preserve sufficient data for future verification.

\subsection{Detailed Watermarking Procedure in Secure FL}
\label{Sec:WatermarkingProcedure}
In this section, we introduce the details of watermarking procedure in Secure FL. Secure FL enhances the protection of privacy by protecting both data and gradients, while watermarking provides copyright protection of FL models. In general, the training procedure of FL proceeds in four phases: initialisation, local training, secure transmission, and aggregation.

\subsubsection{Initialisation}
In the initialisation phase, the initiator uses the initialisation function from~\cite{he2015delving} to initialise the global model and sends it to all clients. Let $M_j \in \mathbb{R}^d$ be the parameters of the model in the $j$-th iteration, where $\mathbb{R}$ is the set of real numbers, and $d$ is the number of model parameters.  With $Init: \varnothing \rightarrow \mathbb{R}^d$ the initialisation function, the initialisation phase is then defines as
\begin{equation}
\label{Eq:Initialization}
M_0 \leftarrow Init(\cdot), M_0 \in \mathbb{R}^d .
\end{equation}

For the initiator, besides getting the initialised global model, the initiator should key generation and trigger set construction algorithms introduced in Section~\ref{Sec:KeyGeneration} and Section~\ref{Sec:TriggerSetConstruction}. The generated trigger set will be used in the following procedure.

\subsubsection{Local Training}
In the local training phase, each client uses their own dataset to train their local model and calculate model gradients to optimise the defined loss function. When $D^i$ is the private dataset of the $i$-th client, $x \in I$ is the input sample, $y \in O$ is the output, and $Loss: \mathbb{R}^d \times \mathbb{D}^{*} \rightarrow \mathbb{R}$ is the loss function, the local gradients $L_j^i$ of the $i$-th client in the $j$-th iteration can be calculated by
\begin{equation}
\label{Eq:LocalTrain}
L_j^i = \eta\frac{\partial Loss(M_{j-1},  D^i)}{\partial M_{j-1}} ,
\end{equation}
where $\eta$ is the learning rate of the clients. In our implementation, we use stochastic gradient descent (SGD) as the optimiser and cross entropy loss as the loss function.

\subsubsection{Secure Transmission}
Before the local gradients to the central server, the clients need to encrypt the plaintext of the gradients to prevent a malicious central server from learning private information from the gradients~\cite{aono2017privacy}. We employ the state-of-the-art HE scheme CKKS (Cheon-Kim-Kim-Song)~\cite{cheon2017homomorphic} to encrypt each gradient. CKKS ensures both privacy and computability so that the gradients can be safely sent to the server for further processing.

For the initiator, Alice should utilize the gradient-enhanced algorithm introduced in Section~\ref{Sec:WatermarkEmbedding} to embed the watermark into the FL model. The enhanced gradients can overcome the limitations of single node in FL and effectively embed the watermark.

\subsubsection{Aggregation}
In the aggregation phase, the central server collects the ciphertexts of the local gradients of $n$ clients and aggregates the gradients to yield the global gradients by  
\begin{equation}
\label{Eq:Aggregation}
G_j =  Agg(\{L_j^0, L_j^1, \cdots ,L_j^{n-1} \}) ,
\end{equation}
where $Agg:\{0, 1\}^{c\times n}\rightarrow \{0, 1\}^c$ is the aggregation function, and $\rho$ is the learning rate of the central server. In our proposed scheme, the server performs the FedAvg~\cite{mcmahan2017communication} aggregation algorithm for this which is defined as
\begin{equation}
FedAvg(\{L_j^0, L_j^1, \cdots ,L_j^{n-1} \})=\frac{1}{\sum_{i=1}^n q^i}\sum_{i=1}^n q^i L_j^i ,
\end{equation}
where $q_i$ refers to the $i$-th client's quantity of its own dataset samples. By employing the CKKS computing method of ciphertexts, we do not need to alter the original aggregation algorithm to apply HE. After aggregation, clients can then use the global gradients to update their models.

\section{Security Analysis}
\label{Sec:SecurityAnalysis}
In security analysis, we focus on two significant properties, non-ambiguity and robustness. 

\subsection{Non-ambiguity}

Non-ambiguity defined in Section~\ref{Sec:DesignGoals} means the watermark can hardly be forged by the attacker. For our proposed scheme, we have the following theorem.

\begin{theorem}
\label{Theorem: Non-ambiguity}
The key generation and trigger set construction algorithms are non-ambiguous for a not so small number of classes in the primitive task, that is, any algorithm $g:\mathbb{R}^d \rightarrow \mathbb{D}^{|D_s|}$ will fail to forge a secret key $sk'$ and trigger set $D_s'$.
\end{theorem}

\begin{proof}
For an adversary who obtains the watermarked model and tries to forge a secret key and trigger set, there are two ways to do that: (1) generate several random patterns and attempt to find the true secret key $sk$, or (2) establish a permutation and use an image generator to construct a trigger set so that he can successfully verify the copyright.

In the first strategy, the adversary tries to find the true secret key by brute force. Thus, the required time consumption and security mainly depend on the secret key space. As discussed in Section~\ref{Sec:TriggerSetConstruction}, the patch parameters $\mu$ and $\nu$ determine the size of the location key space, while the classification key space depends on the number of classes. For a $k$-class model, the length of permutations should be $k$, leading
\begin{equation}
\mathrm{P}_{\mu\nu}^{k}\mathrm{P}_{k}^{k}=\frac{(\mu\nu)!k!}{(\mu\nu - k)!} 
\end{equation}
different secret keys, where $\mathrm{P}$ denotes permutation number and $\mathrm{P}_n^m=\frac{n!}{(n-m)!}$. Therefore, as $\mu\nu$ grows, with the algorithm of $O(k!n^k)$, this is hard to accomplish in practice.

In the second strategy, the adversary needs to generate a trigger set with its own secret key. Assume that for a protected $k$-class model $\hat{M}$, the randomly generated images will also be randomly classified, i.e., the probability of an image to be classified to any class will be $1/k$. Thus, the probability of forging a satisfying trigger set $D_s'$ is 
\begin{equation}
\mathrm{Pr}[Verify(\hat{M}, D_s')\mid D_s' = Random(\cdot)]=\frac{1}{|D_s'|k^k} ,
\end{equation}
which is also not acceptable for some $k$. Considering a simple example which embeds watermarks into the 10-class CIFAR-10 dataset~\cite{krizhevsky2009learning}, we choose $\mu = \nu = 4$ and assume that it takes $0.01$ second for one sample's inference. For a $1000$-image trigger set, it would then take approximately $3170$ years to traverse the secret key space.

Consequently, both strategies to forge a trigger set fail, proving our watermark scheme to be non-ambiguous.
\end{proof}

\subsection{Robustness}
\label{Sec:RobustnessAnalysis}
The robustness of a backdoor-based watermark is related to the strong backdoor defined in Section~\ref{Sec:BackdoorAttack}. That the watermark is robust is equivalent to the backdoor being a strong backdoor. Thus, to guarantee robustness, we need to prove the trigger set embedded using our proposed scheme is the strong backdoor.

\begin{theorem}
For any attack algorithm $Att \in \{ f \mid {\rm dom}f = \mathbb{R}^d \times S \}$, the watermarked model from our watermark embedding method satisfies that either Case~\ref{Prop:Robustness1} or Case~\ref{Prop:Robustness2} is true, i.e., it satisfies the requirement of \emph{robustness}.
\end{theorem}

\begin{proof}
As demonstrated in~\cite{adi2018turning}, when a backdoor is embedded by adding trigger set samples into each batch, the backdoor is a strong backdoor. A model containing a strong backdoor means that this backdoor, and therefore the watermark, cannot be removed by an algorithm without a striking decline of the primitive functionality. As presented in Section~\ref{Sec:WatermarkEmbedding}, the backdoor we embed in our scheme is a strong backdoor which means our scheme satisfies the robustness requirement. 
\end{proof}

Additionally, the theorem can also be confirmed by our experiments in Section~\ref{Sec:Evaluation}.

\section{Evaluation}
\label{Sec:Evaluation}

\setlength{\abovecaptionskip}{0.1cm}
\begin{figure*}[t]
    \centering
    \includegraphics[width=.95\linewidth]{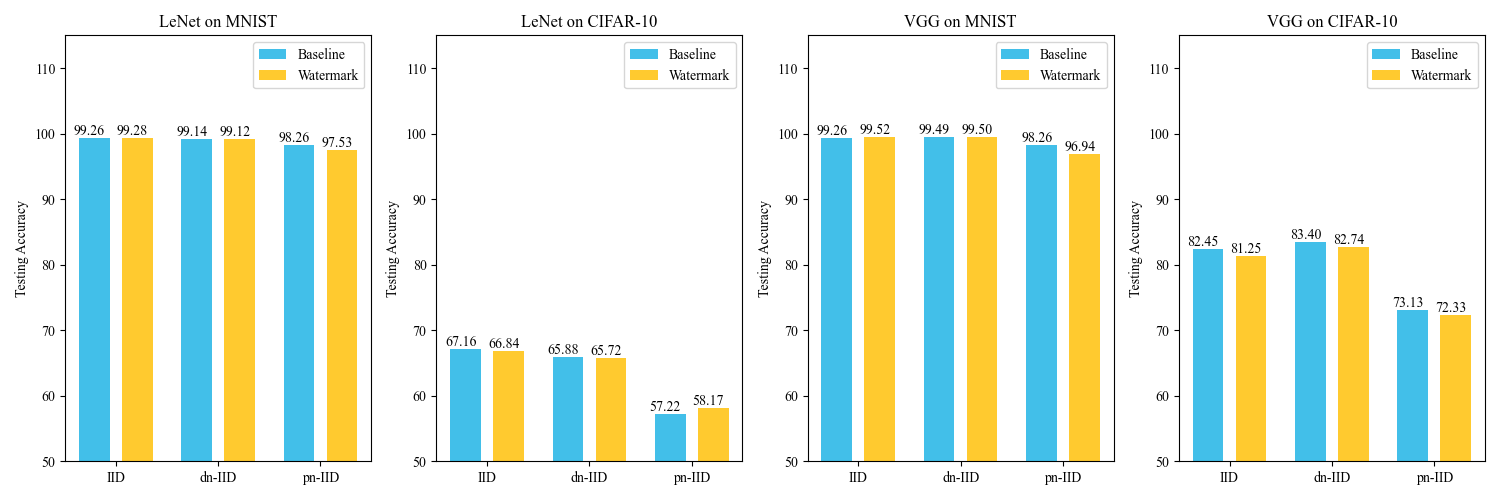}
    \caption{Function preservation results of watermarked models.}
    \label{Fig:Function-preserving}
    \vspace{-15pt}
\end{figure*}

\subsection{Experimental Settings}
In our implementation, we use TensorFlow (version 2.3.4) and TensorFlow-Federated (version 0.17.0) as DL and FL framework, respectively. We encrypt the gradients using Tenseal (version 0.3.4) which implements CKKS. The learning rate of each client is set to $0.01$ while the learning rate of the central server is $1.0$. In the local training phase, each selected client trains the model for two local epochs before transmitting the gradients to central server. The number of clients is set to $100$ to simulate a real-world FL environment.

In our experiments, we adopt two different convolutional neural network (CNN) architectures to evaluate the performance of the proposed method. The first CNN is, as also in~\cite{mcmahan2017communication}, the classical LeNet~\cite{lecun1998gradient} which consists of two convolutional layers, each followed by a $2\times2$ max-pooling layer, and two fully connected layers. The second is VGG~\cite{simonyan2014very} with 13 convolutional layers and 3 fully connected layers.

\subsection{Dataset Settings}
In our experiments, we use two datasets, MNIST~\cite{y2010MNIST} and CIFAR10~\cite{krizhevsky2009learning}. MNIST is a grayscale image dataset with 60,000 training data of the 10 digits and an additional 10,000 for testing, while CIFAR10 comprises 50,000 RGB images for training and 10,000 for testing. To keep consistency, the $28\times28$-pixel images in MNIST are resized to $32\times32$ pixels, the same size as CIFAR10 images.

For data distribution, we split the data using three different distributions. One is the independent identically distribution (IID) which means that the whole training dataset is uniformly allocated to each client. The other two are non-IID distributions. For one, denoted as dn-IID, we use a Dirichlet distribution to split the data accordingly, as used in~\cite{yurochkin2019bayesian}. For the other, we employ the method from~\cite{mcmahan2017communication} to construct a pathological non-IID distribution (pn-IID), which first sorts the dataset according to the labels and then splits them into $tN$ parts where $N$ is the number of clients, with each client randomly choosing $t$ parts as its training set. In this case, each client only has data with $t$ different labels at most. We implement the settings in~\cite{mcmahan2017communication} and set $t=2$ in our experiments.

For trigger set construction, we set the default patch parameters to $\mu = \nu = 4$ for both tasks, making them large enough for security, while in the experiments, for comparing different patch parameters, $\mu=\nu=4, 6, 16$ are used. Before embedding the watermark, we generate 100 images as the trigger set, with each class comprising 10 examples, while we use 1,000 images constructed with the same secret key for verification to test the generalisation ability of our watermarking scheme.

\subsection{Effectiveness}
Effectiveness indicates whether the watermark is successfully embedded into the model, i.e., by the accuracy on the trigger set $D_s$. The effectiveness results are shown in Table~\ref{Table:Effectiveness}.

\begin{table}[h!]
\caption{Watermarking accuracy of watermarked FL models.}
\label{Table:Effectiveness}
\centering
\begin{tabular}{l|ccc|ccc}
\hline
\hline
& \multicolumn{3}{c|}{MNIST} & \multicolumn{3}{c}{CIFAR10} \\
& IID & dn-IID & pn-IID & IID & dn-IID & pn-IID\\
\hline
\hline
LeNet & 1.000 & 1.000 & 0.999 & 0.998 & 1.000 & 0.996\\
VGG & 1.000 & 1.000 & 0.989 & 0.997 & 1.000 & 0.995\\
\hline
\hline
\end{tabular}
\vspace{-10pt}
\end{table}

Table~\ref{Table:Effectiveness} demonstrates the success of embedding the watermark into the FL models. All accuracies on the secret trigger set exceed $98\%$. This also indicates that our trigger set construction algorithm has appropriate generalisation ability for both datasets and both architectures.

\subsection{Function Preservation}
We compare the watermarked FL models with normal models without watermarks in terms of the primitive functionality. The results of this are shown in~Fig.~\ref{Fig:Function-preserving}, from which we can see that the watermarked models do quite well, with a drop in accuracy of below $1\%$ in most cases. This demonstrates that our watermarking scheme has a negligible impact on the functionality of FL models.

\subsection{False Positive Rate}
We take the normal FL models to measure whether the watermark can be detected. In the black-box watermark scheme, this means calculating the accuracy of the secret trigger set with normally-trained models. As demonstrated in Fig.~\ref{Fig:False-positive}, our watermarking scheme satisfies the requirement defined in Section~\ref{Sec:DesignGoals}, with a maximum false positive rate of $11\%$, which means the rates are near or below the expected $10\%$ for a 10-class problem.

\subsection{Robustness}
We conduct four different watermark removal attacks to evaluate the robustness of our watermark, including fine-tuning, pruning, quantisation, and PST attacks as defined in Section~\ref{Sec:ThreatModel}.

\begin{figure}[t!]
    \centering
    \includegraphics[width=.95\columnwidth]{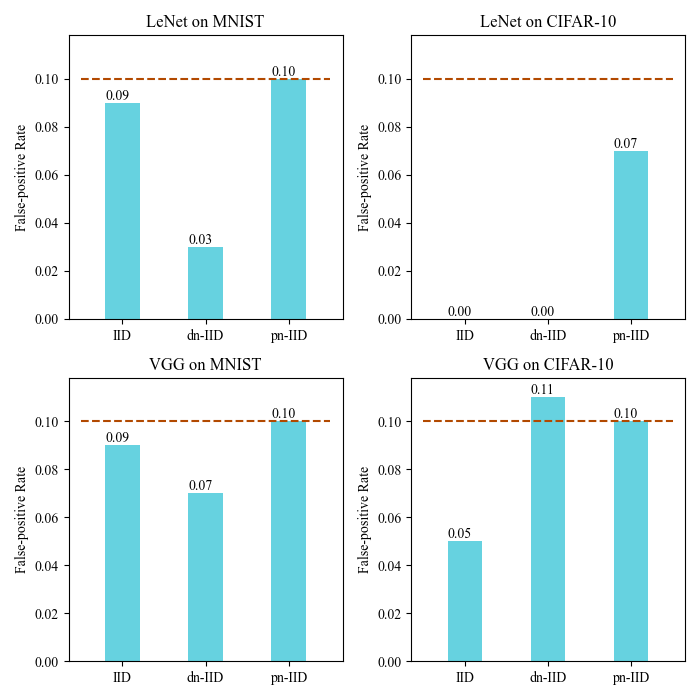}
    \caption{False positive rates of detecting watermark in unwatermarked models.}
    \label{Fig:False-positive}
    \vspace{-4pt}
\end{figure}
    
\begin{figure}[t!]
    \centering
    \includegraphics[width=0.90\columnwidth]{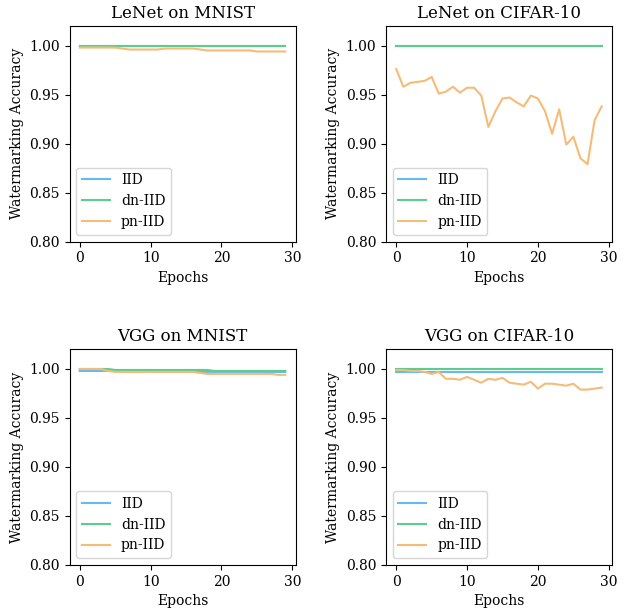}
    \caption{Results of models under fine-tuning attack.}
    \label{Fig:Fine-tuningAttack}
    \vspace{-10pt}
\end{figure}

\subsubsection{Fine-tuning Attack}
In the fine-tuning attack experiment, we fine-tune the watermarked FL model for 30 epochs. As illustrated in Fig.~\ref{Fig:Fine-tuningAttack}, the FL models successfully maintain the watermark for most tasks, with an occasional slight drop in the accuracy of the trigger set. All accuracies remain above $87\%$, sufficient to verify copyright. 

\begin{table*}[h]
    \caption{Watermark (wm) and test accuracy against pruning attack.}
    \label{Table:Pruning}
    \centering
    \begin{tabular}{lc|cc|cc|cc||cc|cc|cc}
    \hline
    \hline
    & & \multicolumn{6}{c||}{MNIST} & \multicolumn{6}{c}{CIFAR10} \\
    &  & \multicolumn{2}{c|}{IID} & \multicolumn{2}{c|}{dn-IID} & \multicolumn{2}{c||}{pn-IID} & \multicolumn{2}{c|}{IID} & \multicolumn{2}{c|}{dn-IID} & \multicolumn{2}{c}{pn-IID}\\
    & pruning rate & wm & test & wm & test & wm & test & wm & test & wm & test & wm & test\\
    \hline
    \hline
    \multirow{10}{*}{LeNet} & 5\% & 1.000& 0.993& 1.000& 0.991& 0.999& 0.975& 0.998& 0.668& 1.000& 0.658& 0.997& 0.582\\
    & 10\% & 1.000& 0.993& 1.000& 0.991& 0.999& 0.975& 0.998& 0.667& 1.000& 0.657& 0.997& 0.580\\
    & 15\% & 1.000& 0.993& 1.000& 0.991& 0.999& 0.975& 0.998& 0.666& 1.000& 0.656& 0.995& 0.582\\
    & 20\% & 1.000& 0.993& 1.000& 0.991& 0.999& 0.974& 0.996& 0.667& 1.000& 0.658& 0.995& 0.584\\
    & 25\% & 1.000& 0.993& 1.000& 0.991& 0.999& 0.975& 0.998& 0.663& 1.000& 0.657& 0.993& 0.574\\
    & 30\% & 1.000& 0.993& 1.000& 0.991& 0.999& 0.974& 0.998& 0.663& 1.000& 0.656& 0.991& 0.569\\
    & 35\% & 1.000& 0.992& 1.000& 0.991& 1.000& 0.974& 0.998& 0.660& 0.999& 0.652& 0.995& 0.560\\
    & 40\% & 1.000& 0.992& 1.000& 0.991& 1.000& 0.973& 0.993& 0.656& 0.999& 0.649& 0.979& 0.553\\
    & 45\% & 1.000& 0.992& 1.000& 0.990& 0.999& 0.975& 0.988& 0.651& 0.998& 0.646& 0.897& 0.541\\
    & 50\% & 1.000& 0.992& 1.000& 0.990& 0.999& 0.973& 0.981& 0.641& 0.996& 0.632& 0.888& 0.527\\
    \hline
    \multirow{10}{*}{VGG} & 5\% & 1.000& 0.995& 1.000& 0.995& 0.989& 0.969& 0.997& 0.812& 1.000& 0.827& 0.996& 0.724\\
    & 10\% & 1.000& 0.995& 1.000& 0.995& 0.988& 0.969& 0.997& 0.812& 1.000& 0.827& 0.995& 0.724\\
    & 15\% & 1.000& 0.995& 1.000& 0.995& 0.988& 0.969& 0.997& 0.811& 1.000& 0.825& 0.997& 0.724\\
    & 20\% & 1.000& 0.995& 1.000& 0.995& 0.988& 0.968& 0.997& 0.810& 1.000& 0.822& 0.997& 0.726\\
    & 25\% & 1.000& 0.995& 1.000& 0.995& 0.987& 0.967& 0.997& 0.808& 1.000& 0.821& 0.997& 0.728\\
    & 30\% & 1.000& 0.995& 1.000& 0.995& 0.987& 0.966& 0.996& 0.801& 1.000& 0.818& 0.996& 0.725\\
    & 35\% & 1.000& 0.995& 1.000& 0.995& 0.983& 0.964& 0.994& 0.791& 1.000& 0.815& 0.996& 0.721\\
    & 40\% & 1.000& 0.995& 0.998& 0.994& 0.934& 0.953& 0.991& 0.773& 1.000& 0.808& 0.988& 0.708\\
    & 45\% & 1.000& 0.994& 0.985& 0.992& 0.669& 0.920& 0.961& 0.743& 1.000& 0.790& 0.961& 0.676\\
    & 50\% & 1.000& 0.993& 0.921& 0.989& 0.191& 0.786& 0.879& 0.683& 1.000& 0.743& 0.779& 0.609\\
    \hline
    \hline
    \end{tabular}
    \vspace{-5pt}
\end{table*}

\begin{table*}[h]
    \caption{Watermark (wm) and test accuracy against quantisation attack.}
    \label{Table:Quantisation}
    \centering
    \begin{tabular}{lc|cc|cc|cc||cc|cc|cc}
    \hline
    \hline
    &  & \multicolumn{6}{c||}{MNIST} & \multicolumn{6}{c}{CIFAR10} \\
    &  & \multicolumn{2}{c|}{IID} & \multicolumn{2}{c|}{dn-IID} & \multicolumn{2}{c||}{pn-IID} & \multicolumn{2}{c|}{IID} & \multicolumn{2}{c|}{dn-IID} & \multicolumn{2}{c}{pn-IID}\\
    & \# bits & wm & test & wm & test & wm & test & wm & test & wm & test & wm & test\\
    \hline
    \hline
    \multirow{7}{*}{LeNet} & 8 bits & 1.000& 0.993& 1.000& 0.991& 0.999& 0.975& 0.998& 0.667& 1.000& 0.657& 0.997& 0.581\\
    & 7 bits & 1.000& 0.993& 1.000& 0.991& 0.999& 0.975& 0.998& 0.668& 1.000& 0.657& 0.997& 0.580\\
    & 6 bits & 1.000& 0.993& 1.000& 0.991& 0.999& 0.975& 0.996& 0.666& 1.000& 0.659& 0.995& 0.580\\
    & 5 bits & 1.000& 0.993& 1.000& 0.991& 0.999& 0.975& 0.992& 0.667& 1.000& 0.656& 0.994& 0.585\\
    & 4 bits & 1.000& 0.993& 1.000& 0.991& 0.998& 0.975& 0.999& 0.656& 1.000& 0.646& 0.984& 0.579\\
    & 3 bits & 0.988& 0.990& 0.997& 0.988& 0.999& 0.974& 0.949& 0.632& 0.953& 0.614& 0.925& 0.531\\
    & 2 bits & 0.999& 0.980& 0.930& 0.974& 0.985& 0.963& 0.140& 0.167& 0.300& 0.240& 0.418& 0.320\\
    \hline
    \multirow{7}{*}{VGG} & 8 bits & 1.000& 0.995& 1.000& 0.995& 0.989& 0.969& 0.997& 0.813& 1.000& 0.827& 1.000& 0.723\\
    & 7 bits & 1.000& 0.995& 1.000& 0.995& 0.989& 0.969& 0.997& 0.812& 1.000& 0.827& 1.000& 0.723\\
    & 6 bits & 1.000& 0.995& 1.000& 0.995& 0.988& 0.969& 0.997& 0.813& 1.000& 0.826& 0.999& 0.726\\
    & 5 bits & 1.000& 0.995& 1.000& 0.995& 0.990& 0.970& 0.997& 0.810& 1.000& 0.827& 0.999& 0.716\\
    & 4 bits & 1.000& 0.995& 1.000& 0.995& 0.989& 0.969& 0.997& 0.798& 1.000& 0.812& 0.997& 0.711\\
    & 3 bits & 1.000& 0.995& 0.997& 0.993& 0.999& 0.966& 0.787& 0.642& 0.991& 0.733& 0.886& 0.642\\
    & 2 bits & 0.100& 0.114& 0.100& 0.114& 0.884& 0.898& 0.100& 0.100& 0.100& 0.100& 0.000& 0.100\\
    \hline
    \hline
    \end{tabular}
    \vspace{-5pt}
\end{table*}

For the FL models trained on IID and dn-IID datasets, the accuracies on the trigger set basically do not change. This is probably because the IID and dn-IID datasets are more manageable tasks for FL~\cite{li2021federated}, and the models have already converged to a reasonable minimum. However, the existing FedAvg algorithm does not perform too well in the pn-IID settings but still acceptable.

\subsubsection{Pruning Attack}
We evaluate the robustness of our watermark against pruning attacks, in particular against parameter pruning~\cite{han2015learning} as conducted in~\cite{uchida2017embedding}. In parameter pruning, parameters close to zero are set to zero, which can intentionally or unintentionally affect the effectiveness of the watermark.

In Table~\ref{Table:Pruning}, we show the pruning attack results for all twelve models and for different percentage levels which indicate how many parameters are pruned. The obtained accuracies on the trigger set $D_s$ are remarkable in almost every task. Most accuracies are above $85\%$, indicating an impressive ability to resist pruning attack. When pruning $50\%$ parameters of VGG trained on the pn-IID MNIST dataset, the accuracy of $D_s$ drops to $0.191\%$. However, at the same time, the accuracy of the benign dataset $D_b$ also drops to about $20\%$, which means the attacked model has lost its value and that therefore we do not regard it as a successful attack.

\begin{figure}[t]
    \centering
    \includegraphics[width=0.95\columnwidth]{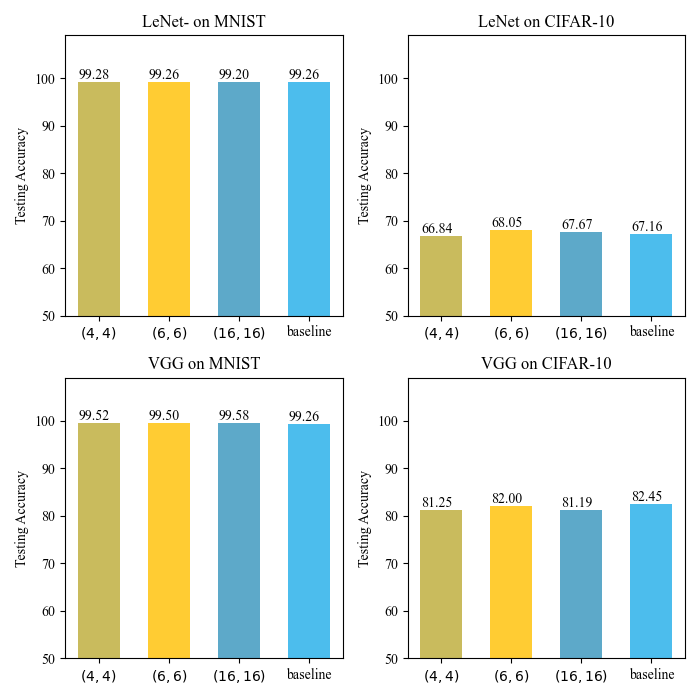}
    \caption{Function preservation results for different patch parameter settings.}
    \label{Fig:DPPFunction-preserving}
    \vspace{-10pt}
\end{figure}

\begin{figure}[t]
    \centering
    \includegraphics[width=1.0\columnwidth]{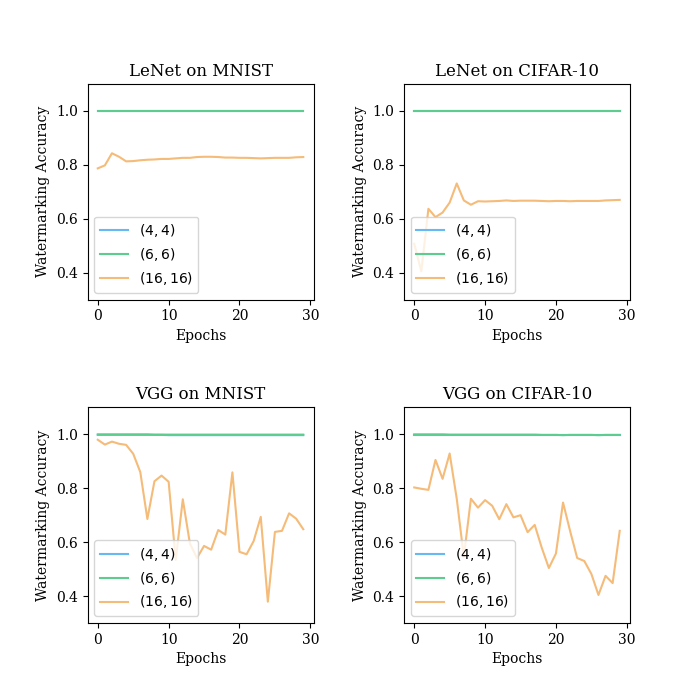}
    \caption{Robustness results of models under fine-tuning attack for different patch parameter settings}
    \label{Fig:DPPfine-tuningattack}
    \vspace{-10pt}
\end{figure}

\begin{table}[t]
    \caption{Robustness of watermark against PST attack.}
    \label{Table:PSTAttack}
    \centering
    \begin{tabular}{l|ccc|ccc}
    \hline
    \hline
    & \multicolumn{3}{c|}{MNIST} & \multicolumn{3}{c}{CIFAR10} \\
    & IID & dn-IID & pn-IID & IID & dn-IID & pn-IID\\
    \hline
    \hline
    LeNet & 0.848 & 0.901 & 0.777 & 0.621 & 0.611 & 0.853\\
    VGG & 0.989 & 0.823 & 0.823 & 0.779 & 0.957 & 0.615\\
    \hline
    \hline
    \end{tabular}
    \vspace{-10pt}
\end{table}

\subsubsection{Quantisation Attack}
In our quantisation attack, we adopt a straightforward method to reduce the weights' bit depth. For each layer in the neural network, we first find the maximum and minimum of the parameters and uniformly divide this interval, and then round parameters to the nearest value. 

We report the accuracy of $D_s$ with $7$ different bit depths in Table~\ref{Table:Quantisation}. As we can see, when the models are quantised to more than 2 bits, the watermarked models preserve both their primitive functionality and the watermark to a good extent, with accuracies on $D_s$ of at least $78.7\%$ and in most cases much higher. For extreme quantisation to 2 bits, some models, especially VGG, lose the ability to classify both $D_b$ and $D_s$, which, as discussed in Section~\ref{Sec:ThreatModel}, would be considered an unsuccessful attack.

\subsubsection{PST Attack}
\label{Sec:PSTAttack}
We implement a PST attack with the same parameters used in~\cite{guo2021fine} and report the results in Table~\ref{Table:PSTAttack}. Only for three of the twelve models, the accuracy drop to below $75\%$, although still above $60\%$. Considering the difficulty of learning pn-IID datasets, this drop might be because of the inadequate capability of LeNet for CIFAR10 or a bad local minimum. The results indicate that PST attacks can have an impact on our watermark to some extent, but are not consistently effective against our trigger set.

\subsection{Evaluation of Different Patch Parameters}
\label{Sec:DifferentPatchParameters}
As discussed in Section~\ref{Sec:TriggerSetConstruction}, the patch parameters can affect the security and effectiveness of our watermarking scheme. We conduct some experiments where we investigate three different settings, $\mu=\nu=4$, $\mu=\nu=6$, and $\mu=\nu=16$, giving corresponding patch sizes of $8\times8$, $5\times5$, and $2\times2$. The $2\times2$ is the minimum patch size to generate sufficient trigger set images. For simplicity, we use $(x, x)$ to represents models with patch parameters $\mu=\nu=x$. Since data distribution is not considered in this section's experiments, the models are trained following the IID setting.

\subsubsection{Effectiveness and Function Preservation}
First, we evaluate the effectiveness and function preservation properties for the three parameter settings and give the results in Table~\ref{Table:DPPEffectiveness} and Fig.~\ref{Fig:DPPFunction-preserving}. From Table \ref{Table:DPPEffectiveness}, although the accuracies are all greater than or close to $99\%$, we can see a slight difference for $(16, 16)$ model. This might be due to overfitting of the $2\times2$ patches. Fig.~\ref{Fig:DPPFunction-preserving} also suggests the same drop. The $(6, 6)$ model yields the best performance, and the $(16, 16)$ model the worst.

\begin{table}[h]
    \caption{Watermarking accuracy for different patch parameter settings}
    \label{Table:DPPEffectiveness}
    \centering
    \begin{tabular}{l|ccc|ccc}
    \hline
    \hline
    & \multicolumn{3}{c|}{MNIST} & \multicolumn{3}{c}{CIFAR10} \\
     & $\mu$=$\nu$=4 & $\mu$=$\nu$=6 & $\mu$=$\nu$=16 & $\mu$=$\nu$=4 & $\mu$=$\nu$=6 & $\mu$=$\nu$=16 \\
    \hline
    \hline
    LeNet & 1.000 & 1.000 & 0.999 & 0.998 & 1.000 & 0.996\\
    VGG & 1.000 & 1.000 & 0.989 & 0.997 & 1.000 & 0.995\\
    \hline
    \hline
    \end{tabular}
    \vspace{-10pt}
\end{table}

\subsubsection{Robustness}
As before, we conduct fine-tuning, pruning, quantisation, and PST attacks, and show the obtained results in Fig.~\ref{Fig:DPPfine-tuningattack}, Table~\ref{Table:DPPPruning}, Table~\ref{Table:DPPQuantisation}, and Table~\ref{Table:DPPPST}, respectively.

\begin{table*}[t]
\caption{Watermarking (wm) and test accuracy against pruning attack for different patch parameter settings.}
\label{Table:DPPPruning}
\centering
\begin{tabular}{lc|cc|cc|cc||cc|cc|cc}
\hline
\hline
& & \multicolumn{6}{c||}{MNIST} & \multicolumn{6}{c}{CIFAR10} \\
&  & \multicolumn{2}{c|}{$\mu$=$\nu$=4} & \multicolumn{2}{c|}{$\mu$=$\nu$=6} & \multicolumn{2}{c||}{$\mu$=$\nu$=16} & \multicolumn{2}{c|}{$\mu$=$\nu$=4} & \multicolumn{2}{c|}{$\mu$=$\nu$=6} & \multicolumn{2}{c}{$\mu$=$\nu$=16}\\
& pruning rate & wm & test & wm & test & wm & test & wm & test & wm & test & wm & test\\
\hline
\hline
\multirow{10}{*}{LeNet} & 5\% & 1.000& 0.993& 1.000& 0.993& 0.923& 0.992& 0.998& 0.668& 1.000& 0.680& 0.990& 0.677\\
& 10\% & 1.000& 0.993& 1.000& 0.993& 0.923& 0.992& 0.998& 0.667& 1.000& 0.680& 0.990& 0.677\\
& 15\% & 1.000& 0.993& 1.000& 0.993& 0.923& 0.992& 0.998& 0.666& 1.000& 0.680& 0.990& 0.677\\
& 20\% & 1.000& 0.993& 1.000& 0.993& 0.922& 0.992& 0.996& 0.667& 1.000& 0.677& 0.988& 0.674\\
& 25\% & 1.000& 0.993& 1.000& 0.993& 0.922& 0.992& 0.998& 0.663& 1.000& 0.678& 0.988& 0.672\\
& 30\% & 1.000& 0.993& 0.999& 0.993& 0.919& 0.992& 0.998& 0.663& 1.000& 0.673& 0.989& 0.669\\
& 35\% & 1.000& 0.992& 0.996& 0.993& 0.922& 0.992& 0.998& 0.660& 0.999& 0.670& 0.984& 0.665\\
& 40\% & 1.000& 0.992& 0.996& 0.993& 0.922& 0.992& 0.993& 0.656& 0.998& 0.666& 0.981& 0.660\\
& 45\% & 1.000& 0.992& 0.995& 0.992& 0.929& 0.991& 0.988& 0.651& 0.995& 0.660& 0.974& 0.652\\
& 50\% & 1.000& 0.992& 0.997& 0.992& 0.930& 0.991& 0.981& 0.641& 0.984& 0.646& 0.950& 0.648\\
\hline
\multirow{10}{*}{VGG} & 5\% & 1.000& 0.995& 1.000& 0.995& 0.966& 0.996& 0.997& 0.812& 1.000& 0.820& 0.999& 0.813\\
& 10\% & 1.000& 0.995& 1.000& 0.995& 0.966& 0.996& 0.997& 0.812& 1.000& 0.821& 1.000& 0.813\\
& 15\% & 1.000& 0.995& 1.000& 0.995& 0.966& 0.996& 0.997& 0.811& 1.000& 0.819& 0.999& 0.811\\
& 20\% & 1.000& 0.995& 1.000& 0.995& 0.966& 0.996& 0.997& 0.810& 1.000& 0.819& 1.000& 0.810\\
& 25\% & 1.000& 0.995& 1.000& 0.995& 0.965& 0.996& 0.997& 0.808& 1.000& 0.820& 0.999& 0.808\\
& 30\% & 1.000& 0.995& 1.000& 0.995& 0.964& 0.996& 0.996& 0.801& 1.000& 0.812& 0.999& 0.800\\
& 35\% & 1.000& 0.995& 0.999& 0.994& 0.965& 0.996& 0.994& 0.791& 0.999& 0.804& 0.999& 0.781\\
& 40\% & 1.000& 0.995& 0.999& 0.994& 0.968& 0.996& 0.991& 0.773& 0.999& 0.798& 0.993& 0.753\\
& 45\% & 1.000& 0.994& 0.999& 0.994& 0.967& 0.995& 0.961& 0.743& 0.989& 0.770& 0.985& 0.699\\
& 50\% & 1.000& 0.993& 0.999& 0.991& 0.959& 0.994& 0.879& 0.683& 0.965& 0.727& 0.900& 0.620\\
\hline
\hline
\end{tabular}
\vspace{-5pt}
\end{table*}

\begin{table*}[t]
\caption{Watermarking (wm) and test accuracy against quantisation attack for different patch parameter settings.}
\label{Table:DPPQuantisation}
\centering
\begin{tabular}{lc|cc|cc|cc||cc|cc|cc}
\hline
\hline
& & \multicolumn{6}{c||}{MNIST} & \multicolumn{6}{c}{CIFAR10} \\
& & \multicolumn{2}{c|}{$\mu$=$\nu$=4} & \multicolumn{2}{c|}{$\mu$=$\nu$=6} & \multicolumn{2}{c||}{$\mu$=$\nu$=16} & \multicolumn{2}{c|}{$\mu$=$\nu$=4} & \multicolumn{2}{c|}{$\mu$=$\nu$=6} & \multicolumn{2}{c}{$\mu$=$\nu$=16}\\
& \# bits & wm & test & wm & test & wm & test & wm & test & wm & test & wm & test\\
\hline
\hline
\multirow{7}{*}{LeNet} & 8 bits & 1.000& 0.993& 1.000& 0.993& 0.968& 0.992& 0.998& 0.667& 1.000& 0.680& 0.998& 0.678\\
& 7 bits & 1.000& 0.993& 1.000& 0.993& 0.969& 0.992& 0.998& 0.668& 1.000& 0.679& 0.997& 0.677\\
& 6 bits & 1.000& 0.993& 1.000& 0.993& 0.967& 0.992& 0.996& 0.666& 1.000& 0.679& 0.998& 0.675\\
& 5 bits & 1.000& 0.993& 0.999& 0.993& 0.967& 0.992& 0.992& 0.667& 1.000& 0.677& 0.995& 0.674\\
& 4 bits & 1.000& 0.993& 0.998& 0.993& 0.960& 0.992& 0.999& 0.656& 1.000& 0.666& 0.991& 0.665\\
& 3 bits & 0.988& 0.990& 0.996& 0.991& 0.960& 0.991& 0.949& 0.632& 0.850& 0.629& 0.468& 0.451\\
& 2 bits & 0.999& 0.980& 0.860& 0.980& 0.864& 0.846& 0.140& 0.167& 0.100& 0.100& 0.105& 0.136\\
\hline
\multirow{7}{*}{VGG} & 8 bits & 1.000& 0.993& 1.000& 0.995& 0.966& 0.996& 0.997& 0.813& 1.000& 0.820& 1.000& 0.812\\
& 7 bits & 1.000& 0.993& 1.000& 0.995& 0.966& 0.996& 0.997& 0.812& 1.000& 0.820& 1.000& 0.812\\
& 6 bits & 1.000& 0.993& 1.000& 0.995& 0.966& 0.996& 0.997& 0.813& 1.000& 0.820& 0.999& 0.810\\
& 5 bits & 1.000& 0.993& 1.000& 0.995& 0.965& 0.996& 0.997& 0.810& 1.000& 0.815& 1.000& 0.802\\
& 4 bits & 1.000& 0.993& 1.000& 0.995& 0.965& 0.996& 0.997& 0.798& 1.000& 0.802& 0.998& 0.800\\
& 3 bits & 0.988& 0.990& 1.000& 0.994& 0.964& 0.995& 0.787& 0.642& 0.845& 0.495& 0.990& 0.596\\
& 2 bits & 0.999& 0.980& 0.100& 0.114& 0.100& 0.114& 0.100& 0.100& 0.100& 0.100& 0.100& 0.100\\
\hline
\hline
\end{tabular}
\vspace{-5pt}
\end{table*}

\begin{table}[t]
\caption{Robustness of watermark against PST Attack for different patch parameter settings.}
\label{Table:DPPPST}
\centering
\begin{tabular}{l|ccc|ccc}
\hline
\hline
& \multicolumn{3}{c|}{MNIST} & \multicolumn{3}{c}{CIFAR10} \\
 & $\mu$=$\nu$=4 & $\mu$=$\nu$=6 & $\mu$=$\nu$=16 & $\mu$=$\nu$=4 & $\mu$=$\nu$=6 & $\mu$=$\nu$=16 \\
\hline
\hline
LeNet & 0.848 & 0.947 & 0.676 & 0.621 & 0.691 & 0.413\\
VGG & 0.989 & 0.804 & 0.544 & 0.779 & 0.779 & 0.529\\
\hline
\hline
\end{tabular}
\vspace{-10pt}
\end{table}

For fine-tuning attacks, from Fig.~\ref{Fig:DPPfine-tuningattack} we can notice a significant impact of overfitting on the watermark. While the $(4, 4)$ and $(6, 6)$ models perform well, maintaining a high accuracy on the trigger set, the accuracy of the $(16,16)$ models drops quite severely. The $2\times2$ patches here are so small that the models have to overfit to classify them, thus affecting the accuracy on the trigger set.

For the pruning and quantisation attacks, Tables~\ref{Table:DPPPruning} and~\ref{Table:DPPQuantisation} show that $(4, 4)$ and $(6, 6)$ provide similarly high robustness, while for $(16, 16)$ we can notice a $5\%$ drop in accuracy. Since these attacks do not significantly alter the decision boundaries of DL models, the adverse effect of overfitting is not evident here.

For the PST attack, from Table~\ref{Table:DPPPST} it is evident that the $(16, 16)$ models cannot resist this type of attack, with accuracies on the trigger set declined by more than $20\%$. The median filter and affine transformations used in PST attacks, lead to severe damage to the small patches because small patches are more easily filtered and transformed. Both the $(4, 4)$ and $(6, 6)$ models exhibit good robustness against PST attacks.

Overall, although $(16, 16)$ gives the largest secret key space, this setting leads to robustness deficiencies since the resulting patch size is too small. $(4, 4)$ and $(6, 6)$ yield similar results in the robustness experiments, and thus they are recommended in our tasks.

\section{Conclusions}
\label{Sec:Conclusions}
In this paper, we have proposed a novel client-side watermarking scheme for homomorphic-encryption-based secure federated learning. To our best knowledge, this is the first scheme to embed the watermark to models in a secure FL environment. The advantages of our scheme are as follows: (1) Using the gradient enhancement method, a client side can embed backdoor-based watermark into the secure FL model; (2) The proposed non-ambiguous trigger set construction mechanism means that an adversary cannot forge the watermark and claim the copyright of the model; (3) The proposed gradient-enhanced watermark embedding method tackles the issue of slim effects of single client on watermark embedding in the FL environment; (4) Using our proposed scheme, the FL model meets the requirements of effectiveness, function preservation, low false positive rate, and resistance to typical watermark removal attacks. In future work, we plan to deploy our watermarking framework into real-life secure FL applications, such as hospitals and credit systems.


%

\ifCLASSOPTIONcompsoc
  \section*{Acknowledgments}
\else
  \section*{Acknowledgment}
\fi

This research is supported by the National Natural Science Foundation of China (61602527,U1734208), Natural Science Foundation of Hunan Province, China (2020JJ4746), and in part by the High Performance Computing Center of Central South University.

\ifCLASSOPTIONcaptionsoff
  \newpage
\fi



%

\bibliography{reference}

\begin{thebibliography}{10}
\providecommand{\url}[1]{#1}
\csname url@samestyle\endcsname
\providecommand{\newblock}{\relax}
\providecommand{\bibinfo}[2]{#2}
\providecommand{\BIBentrySTDinterwordspacing}{\spaceskip=0pt\relax}
\providecommand{\BIBentryALTinterwordstretchfactor}{4}
\providecommand{\BIBentryALTinterwordspacing}{\spaceskip=\fontdimen2\font plus
\BIBentryALTinterwordstretchfactor\fontdimen3\font minus
  \fontdimen4\font\relax}
\providecommand{\BIBforeignlanguage}[2]{{%
\expandafter\ifx\csname l@#1\endcsname\relax
\typeout{** WARNING: IEEEtran.bst: No hyphenation pattern has been}%
\typeout{** loaded for the language `#1'. Using the pattern for}%
\typeout{** the default language instead.}%
\else
\language=\csname l@#1\endcsname
\fi
#2}}
\providecommand{\BIBdecl}{\relax}
\BIBdecl

\bibitem{mcmahan2017communication}
B.~McMahan, E.~Moore, D.~Ramage, S.~Hampson, and B.~A. y~Arcas,
  ``Communication-efficient learning of deep networks from decentralized
  data,'' in \emph{Proceedings of 2017 International Conference Artificial
  Intelligence and Statistics}, 2017, pp. 1273--1282.

\bibitem{adnan2022federated}
M.~Adnan, S.~Kalra, J.~C. Cresswell, G.~W. Taylor, and H.~R. Tizhoosh,
  ``Federated learning and differential privacy for medical image analysis,''
  \emph{Scientific Reports}, vol.~12, no.~1, pp. 1--10, 2022.

\bibitem{ng2021federated}
D.~Ng, X.~Lan, M.~M.-S. Yao, W.~P. Chan, and M.~Feng, ``Federated learning: a
  collaborative effort to achieve better medical imaging models for individual
  sites that have small labelled datasets,'' \emph{Quantitative Imaging in
  Medicine and Surgery}, vol.~11, no.~2, p. 852, 2021.

\bibitem{kumar2021blockchain}
R.~Kumar, A.~A. Khan, J.~Kumar, A.~Zakria, N.~A. Golilarz, S.~Zhang, Y.~Ting,
  C.~Zheng, and W.~Wang, ``Blockchain-federated-learning and deep learning
  models for covid-19 detection using {CT} imaging,'' \emph{IEEE Sensors
  Journal}, 2021.

\bibitem{hard2018federated}
A.~Hard, K.~Rao, R.~Mathews, S.~Ramaswamy, F.~Beaufays, S.~Augenstein,
  H.~Eichner, C.~Kiddon, and D.~Ramage, ``Federated learning for mobile
  keyboard prediction,'' \emph{arXiv preprint arXiv:1811.03604}, 2018.

\bibitem{zhu2020empirical}
X.~Zhu, J.~Wang, Z.~Hong, and J.~Xiao, ``Empirical studies of institutional
  federated learning for natural language processing,'' in \emph{Proceedings of
  2020 Conference on Empirical Methods in Natural Language Processing:
  Findings}, 2020, pp. 625--634.

\bibitem{singhal2021federated}
K.~Singhal, H.~Sidahmed, Z.~Garrett, S.~Wu, J.~Rush, and S.~Prakash,
  ``Federated reconstruction: Partially local federated learning,''
  \emph{Advances in Neural Information Processing Systems}, vol.~34, 2021.

\bibitem{yang2020federated}
L.~Yang, B.~Tan, V.~W. Zheng, K.~Chen, and Q.~Yang, ``Federated recommendation
  systems,'' in \emph{Federated Learning}.\hskip 1em plus 0.5em minus
  0.4em\relax Springer, 2020, pp. 225--239.

\bibitem{muhammad2020fedfast}
K.~Muhammad, Q.~Wang, D.~O'Reilly-Morgan, E.~Tragos, B.~Smyth, N.~Hurley,
  J.~Geraci, and A.~Lawlor, ``Fedfast: Going beyond average for faster training
  of federated recommender systems,'' in \emph{Proceedings of 2020 ACM SIGKDD
  International Conference on Knowledge Discovery \& Data Mining}, 2020, pp.
  1234--1242.

\bibitem{regazzoni2021protecting}
F.~Regazzoni, P.~Palmieri, F.~Smailbegovic, R.~Cammarota, and I.~Polian,
  ``Protecting artificial intelligence {IPs}: A survey of watermarking and
  fingerprinting for machine learning,'' \emph{CAAI Transactions on
  Intelligence Technology}, vol.~6, no.~2, pp. 180--191, 2021.

\bibitem{fkirin2022copyright}
A.~Fkirin, G.~Attiya, A.~El-Sayed, and M.~A. Shouman, ``Copyright protection of
  deep neural network models using digital watermarking: a comparative study,''
  \emph{Multimedia Tools and Applications}, pp. 1--15, 2022.

\bibitem{uchida2017embedding}
Y.~Uchida, Y.~Nagai, S.~Sakazawa, and S.~Satoh, ``Embedding watermarks into
  deep neural networks,'' in \emph{Proceedings of 2017 ACM International
  Conference on Multimedia Retrieval}, 2017, pp. 269--277.

\bibitem{wang2021riga}
T.~Wang and F.~Kerschbaum, ``Riga: Covert and robust white-box watermarking of
  deep neural networks,'' in \emph{Proceedings of 2021 Web Conference}, 2021,
  pp. 993--1004.

\bibitem{szyller2021dawn}
S.~Szyller, B.~G. Atli, S.~Marchal, and N.~Asokan, ``Dawn: Dynamic adversarial
  watermarking of neural networks,'' in \emph{Proceedings of 2021 ACM
  International Conference on Multimedia}, 2021, pp. 4417--4425.

\bibitem{maini2021dataset}
P.~Maini, M.~Yaghini, and N.~Papernot, ``Dataset {{Inference}}: {{Ownership
  Resolution}} in {{Machine Learning}},'' in \emph{Proceedings of 2021
  {{International Conference}} on {{Learning Representations}}}, Apr. 2021.

\bibitem{gu2017badnets}
T.~Gu, B.~Dolan-Gavitt, and S.~Garg, ``Badnets: Identifying vulnerabilities in
  the machine learning model supply chain,'' \emph{arXiv preprint
  arXiv:1708.06733}, 2017.

\bibitem{adi2018turning}
Y.~Adi, C.~Baum, M.~Cisse, B.~Pinkas, and J.~Keshet, ``Turning your weakness
  into a strength: Watermarking deep neural networks by backdooring,'' in
  \emph{Proceedings of 2018 USENIX Security Symposium}, 2018, pp. 1615--1631.

\bibitem{zhang2018protecting}
J.~Zhang, Z.~Gu, J.~Jang, H.~Wu, M.~P. Stoecklin, H.~Huang, and I.~Molloy,
  ``Protecting intellectual property of deep neural networks with
  watermarking,'' in \emph{Proceedings of 2018 Asia Conference on Computer and
  Communications Security}, 2018, pp. 159--172.

\bibitem{merrer2020adversarial}
E.~L. Merrer, P.~Perez, and G.~Tr{\'e}dan, ``Adversarial frontier stitching for
  remote neural network watermarking,'' \emph{Neural Computing and
  Applications}, vol.~32, no.~13, pp. 9233--9244, 2020.

\bibitem{xue2022active}
M.~Xue, S.~Sun, Y.~Zhang, J.~Wang, and W.~Liu, ``Active intellectual property
  protection for deep neural networks through stealthy backdoor and users’
  identities authentication,'' \emph{Applied Intelligence}, pp. 1--15, 2022.

\bibitem{shafieinejad2021robustness}
M.~Shafieinejad, N.~Lukas, J.~Wang, X.~Li, and F.~Kerschbaum, ``On the
  {{Robustness}} of {{Backdoor-based Watermarking}} in {{Deep Neural
  Networks}},'' in \emph{Proceedings of the 2021 {{ACM Workshop}} on
  {{Information Hiding}} and {{Multimedia Security}}}.\hskip 1em plus 0.5em
  minus 0.4em\relax {Virtual Event Belgium}: {ACM}, Jun. 2021, pp. 177--188.

\bibitem{tekgul2021waffle}
B.~G. Tekgul, Y.~Xia, S.~Marchal, and N.~Asokan, ``{WAFFLE}: Watermarking in
  federated learning,'' in \emph{Proceedings of 2021 International Symposium on
  Reliable Distributed Systems}, 2021, pp. 310--320.

\bibitem{aono2017privacy}
Y.~Aono, T.~Hayashi, L.~Wang, S.~Moriai \emph{et~al.}, ``Privacy-preserving
  deep learning via additively homomorphic encryption,'' \emph{IEEE
  Transactions on Information Forensics and Security}, vol.~13, no.~5, pp.
  1333--1345, 2017.

\bibitem{paillier1999public}
P.~Paillier, ``Public-key cryptosystems based on composite degree residuosity
  classes,'' in \emph{Proceedings of 1999 International Conference on the
  Theory and Applications of Cryptographic Techniques}, 1999, pp. 223--238.

\bibitem{yang2019federated}
Q.~Yang, Y.~Liu, T.~Chen, and Y.~Tong, ``Federated machine learning: Concept
  and applications,'' \emph{ACM Transactions on Intelligent Systems and
  Technology}, vol.~10, no.~2, pp. 1--19, 2019.

\bibitem{liu2020reflection}
Y.~Liu, X.~Ma, J.~Bailey, and F.~Lu, ``Reflection backdoor: A natural backdoor
  attack on deep neural networks,'' in \emph{Proceedings of 2020 European
  Conference on Computer Vision}, 2020, pp. 182--199.

\bibitem{bagdasaryan2020backdoor}
E.~Bagdasaryan, A.~Veit, Y.~Hua, D.~Estrin, and V.~Shmatikov, ``How to backdoor
  federated learning,'' in \emph{Proceedings of 2020 International Conference
  on Artificial Intelligence and Statistics}, 2020, pp. 2938--2948.

\bibitem{rieger2022deepsight}
P.~Rieger, T.~D. Nguyen, M.~Miettinen, and A.-R. Sadeghi, ``{{DeepSight}}:
  {{Mitigating Backdoor Attacks}} in {{Federated Learning Through Deep Model
  Inspection}},'' in \emph{Proceedings of 2022 {{Network}} and {{Distributed
  System Security Symposium}}}, Jan. 2022.

\bibitem{guo2018watermarking}
J.~Guo and M.~Potkonjak, ``Watermarking deep neural networks for embedded
  systems,'' in \emph{Proceedings of 2018 IEEE/ACM International Conference on
  Computer-Aided Design}, 2018, pp. 1--8.

\bibitem{li2019how}
Z.~Li, C.~Hu, Y.~Zhang, and S.~Guo, ``How to prove your model belongs to you: A
  blind-watermark based framework to protect intellectual property of {DNN},''
  in \emph{Proceedings of 2019 Annual Computer Security Applications
  Conference}, 2019, pp. 126--137.

\bibitem{rivest1978data}
R.~L. Rivest, L.~Adleman, M.~L. Dertouzos \emph{et~al.}, ``On data banks and
  privacy homomorphisms,'' \emph{Foundations of secure computation}, vol.~4,
  no.~11, pp. 169--180, 1978.

\bibitem{yang2019federatedlearning}
Q.~Yang, Y.~Liu, Y.~Cheng, Y.~Kang, T.~Chen, and H.~Yu, ``Federated learning,''
  \emph{Synthesis Lectures on Artificial Intelligence and Machine Learning},
  vol.~13, no.~3, pp. 1--207, 2019.

\bibitem{park2022privacy}
J.~Park and H.~Lim, ``Privacy-preserving federated learning using homomorphic
  encryption,'' \emph{Applied Sciences}, vol.~12, no.~2, p. 734, 2022.

\bibitem{ma2022privacy}
J.~Ma, S.-A. Naas, S.~Sigg, and X.~Lyu, ``Privacy-preserving federated learning
  based on multi-key homomorphic encryption,'' \emph{International Journal of
  Intelligent Systems}, 2022.

\bibitem{guo2021fine}
S.~Guo, T.~Zhang, H.~Qiu, Y.~Zeng, T.~Xiang, and Y.~Liu, ``Fine-tuning is not
  enough: A simple yet effective watermark removal attack for {DNN} models,''
  in \emph{Proceedings of 2021 International Joint Conference on Artificial
  Intelligence}, 2021.

\bibitem{he2015delving}
K.~He, X.~Zhang, S.~Ren, and J.~Sun, ``Delving deep into rectifiers: Surpassing
  human-level performance on {I}mage{N}et classification,'' in
  \emph{Proceedings of 2015 IEEE international conference on computer vision},
  2015, pp. 1026--1034.

\bibitem{cheon2017homomorphic}
J.~H. Cheon, A.~Kim, M.~Kim, and Y.~Song, ``Homomorphic encryption for
  arithmetic of approximate numbers,'' in \emph{Proceedings of 2017
  International Conference on the Theory and Application of Cryptology and
  Information Security}, 2017, pp. 409--437.

\bibitem{krizhevsky2009learning}
A.~Krizhevsky, G.~Hinton \emph{et~al.}, ``Learning multiple layers of features
  from tiny images,'' \emph{Tech. Rep.}, 2009.

\bibitem{lecun1998gradient}
Y.~LeCun, L.~Bottou, Y.~Bengio, and P.~Haffner, ``Gradient-based learning
  applied to document recognition,'' in \emph{Proceedings of the IEEE},
  vol.~86, no.~11, 1998, pp. 2278--2324.

\bibitem{simonyan2014very}
K.~Simonyan and A.~Zisserman, ``Very deep convolutional networks for
  large-scale image recognition,'' \emph{arXiv preprint arXiv:1409.1556}, 2014.

\bibitem{y2010MNIST}
\BIBentryALTinterwordspacing
L.~Y, C.~C, and B.~C, ``{MNIST} handwritten digit database,'' 2010. [Online].
  Available: \url{http://yann.lecun.com/exdb/mnist/}
\BIBentrySTDinterwordspacing

\bibitem{yurochkin2019bayesian}
M.~Yurochkin, M.~Agarwal, S.~Ghosh, K.~Greenewald, N.~Hoang, and Y.~Khazaeni,
  ``Bayesian nonparametric federated learning of neural networks,'' in
  \emph{Proceedings of 2019 International Conference on Machine Learning},
  2019, pp. 7252--7261.

\bibitem{li2021federated}
Q.~Li, Y.~Diao, Q.~Chen, and B.~He, ``Federated learning on non-iid data silos:
  An experimental study,'' \emph{arXiv preprint arXiv:2102.02079}, 2021.

\bibitem{han2015learning}
S.~Han, J.~Pool, J.~Tran, and W.~Dally, ``Learning both weights and connections
  for efficient neural network,'' in \emph{Advances in Neural Information
  Processing Systems}, vol.~28, 2015.

\end{thebibliography}

\end{document}